\documentclass[11pt]{article}
\usepackage{fullpage}
\pdfoutput=1
\usepackage{amsmath, amsthm, amssymb, bm}
\usepackage{cite}
\usepackage{graphicx}
\usepackage{epstopdf}
\usepackage{multirow}
\usepackage{url}
\usepackage{hyperref, color} 
\usepackage{subfigure} 

\newtheorem{defn}{Definition}
\newtheorem{thm}{Theorem}
\newtheorem{lem}{Lemma}
\newtheorem{cor}{Corollary}
\newtheorem{prop}{Proposition}

\newcommand{\vm}[1]{\boldsymbol{#1}}
\renewcommand{\P}[1]{\operatorname{P}\left\{#1\right\}}
\newcommand{\E}{\operatorname{E}}
\newcommand{\tr}[1]{\operatorname{Tr}\left(#1\right)}
\newcommand{\rank}[1]{\operatorname{rank}\left(#1\right)}
\DeclareMathOperator{\vc}{vec}
\newcommand{\argmin}{\operatorname{argmin}}

\newcommand{\T}{{\operatorname{T}}}
\newcommand{\F}{{\operatorname{F}}}

\newcommand{\er}{\mathrm{e}}
\renewcommand{\j}{\mathrm{j}}
\newcommand{\vct}[1]{\bm{#1}}

\def\y{\vm{Y}}
\def\x{\vm{X}}
\def\u{\boldsymbol{U}}
\def\v{\boldsymbol{V}}
\def\f{\vm{F}}

\def\A{\vm{A}}
\def\z{\vm{Z}}
\def\w{\vm{W}}

\def\k{\kappa}
\def\ff{\vm{f}}
\def\O{\Omega}
\def\o{n}
\def\d{\vm{d}}
\def\op{\nu}

\newcommand{\R}{\mathbb{R}}
\newcommand{\C}{\mathbb{C}}
\newcommand{\uu}{\vct{u}}
\newcommand{\vv}{\vct{v}}

\def\PT{\mathcal{P}_T}
\def\PTc{\mathcal{P}_{T^\perp}}
\def\cA{\mathcal{A}}
\def\cB{\mathcal{B}}
\def\cL{\mathcal{L}}
\def\reals{\mathbb{R}}
\def\comps{\mathbb{C}}

\newcommand{\<}{\langle}
\renewcommand{\>}{\rangle}%----------------------------------------------------------------------------------------------------------------------------------
 
\begin{document}

\title{Compressive Multiplexing of Correlated Signals}

\author{Ali Ahmed and Justin Romberg\thanks{School of Electrical and Computer Engineering, Georgia Tech, Atlanta, GA.  Email: aliahmed@uet.edu.pk, jrom@ece.gatech.edu.  This work was supported by ONR grant N00014-11-1-0459 and NSF grant CNS-0910592.  Submitted to the IEEE Transactions on Information Theory on August 22, 2013.}}

\date{\today}
\maketitle
\begin{abstract}

We present a general architecture for the acquisition of ensembles of correlated signals.  The signals are multiplexed onto a single line by mixing each one against a different code and then adding them together, and the resulting signal is sampled at a high rate.  We show that if the $M$ signals, each bandlimited to $W/2$ Hz, can be approximated by a superposition of $R < M$ underlying signals, then the ensemble can be recovered by sampling at a rate within a logarithmic factor of $RW$ (as compared to the cumulative Nyquist rate of $MW$).  This sampling theorem shows that the correlation structure of the signal ensemble can be exploited in the acquisition process even though it is unknown a priori.

The reconstruction of the ensemble is recast as a low-rank matrix recovery problem from linear measurements.  The architectures we are considering impose a certain type of structure on the linear operators.  Although our results depend on the mixing forms being random, this imposed structure results in a very different type of random projection than those analyzed in the low-rank recovery literature to date.
\end{abstract}
%------------------------------------------------------------------------------------------------------------------------

\section{Introduction}
In this paper, we propose and analyze two multiplexing architectures for the sub-Nyquist acquisition of ensembles of correlated signals. The problem is illustrated in Figure~\ref{fig:CM-M-Mux}: $M$ signals, each of which is bandlimited to $W/2$ radians/sec, are outputs from different sensors.  Our goal is to combine this ensemble into a single signal 
which is then sampled with a standard analog-to-digital converter (ADC).  A conventional way of combining the signals is to use a frequency multiplexer: the signals are modulated to different frequency bands of size $W$ by pre-multiplying them by sinusoids at different frequencies before they are combined.  The signals occupy disjoint bands inside of this combination, so they can be easily separated, and the combined signal has a total bandwidth of $MW/2$ so it can be sampled at  $MW$ samples per second.  Alternatively,  the signals might be time multiplexed at the input of the ADC, again resulting in an overall sampling rate  of $MW$.

\begin{figure}[ht]
  \begin{center}
    \includegraphics[trim=5.5cm 5.5cm 5cm 5cm, clip = true, scale = 0.7]{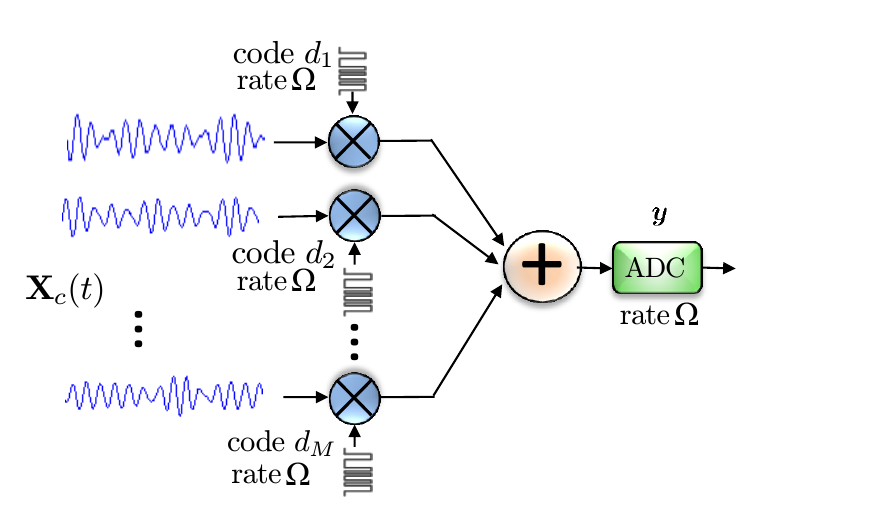}
    \end{center}
    \caption{\small\sl The M-Mux for the efficient acquisition of correlated ensembles.  Signals $\{x_m(t)\}_{1 \leq m \leq M}$ in the ensemble $\x_c(t)$ are multiplied by independently generated random binary waveform $d_1(t), d_2(t), \ldots, d_M(t)$, respectively. The binary waveforms alternate at rate $\Omega$. After the modulation the signals are added and sampled at rate $\Omega$. The reconstruction algorithm uses the nuclear-norm minimization.}
  \label{fig:CM-M-Mux}
\end{figure}

We will show that if the signals are correlated, meaning that the ensemble can be written as (or closely approximated by) distinct linear combinations of $R\ll M$ latent signals, then this net sampling rate can be reduced considerably using random modulators, where the signals are pre-multiplied against random binary waveforms before they are combined.  The multiplexed sampling architectures, we propose are blind to the {\em correlation structure} of the signals; this structure is discovered as the signals are reconstructed.

We recast the problem of recovering the signal ensemble as recovering a low-rank matrix from an incomplete set of linear measurements.  Over the course of one second, we want to acquire an $M\times W$ matrix comprised of samples of the ensemble taken at the Nyquist rate (see Figures~\ref{fig:CM-Fig1} and \ref{fig:CM-Fig2}), and each sample the ADC outputs in this time frame can be written as a different linear combination of the entries in this matrix.  The conditions (on the signals and the acquisition system) under which this type of recovery is effective have undergone intensive study in the recent literature \cite{fazel02ma,recht10gu,candes09ex,keshavan10ma,fazel2008compressed,gross11re,candes10ma}.  
The main contribution of this paper is to show that similar recovery guarantees can be made for measurements with the type of {\em structured randomness} imposed by our multiplexing architecture.  In the context of signal processing, Theorems~\ref{thm:CM-exactrec-CM1},~\ref{thm:CM-stablerec-CM1}, and \ref{thm:CM-exactstablerec-CM2} in Sections~\ref{sec:sampling-thm-m-mux} and \ref{sec:sampling-thm-fm-mux} below provide new sampling theorems for ensembles of correlated signals; in the context of linear algebra, they demonstrate that a low-rank matrix can be recovered from a new kind of low-dimensional random projection whose structure allows efficient computation.

\begin{figure}[htp]
  \begin{center}
	\includegraphics[height=2in]{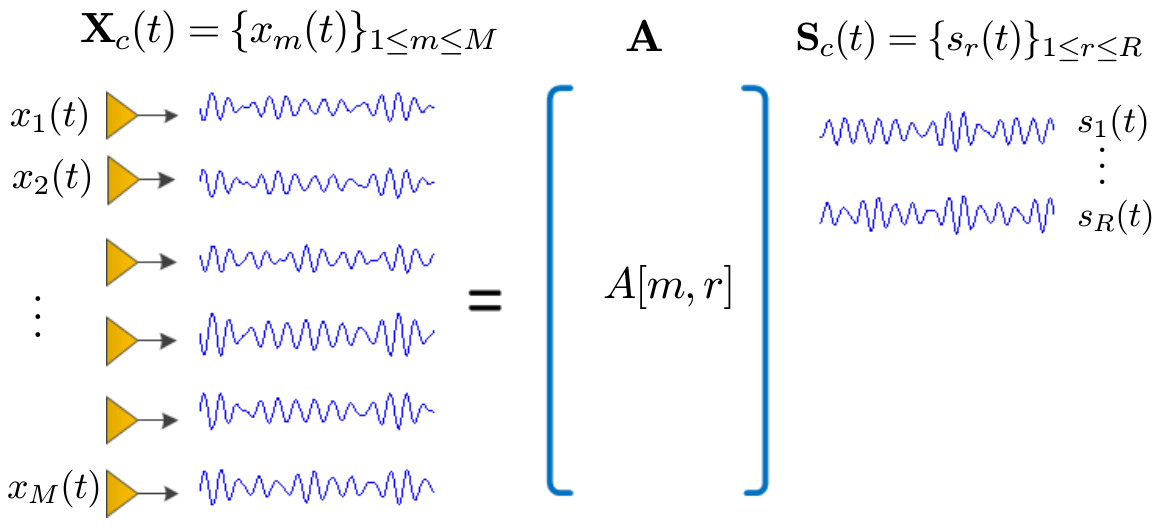}
    \end{center}
    \caption{\small\sl Acquire an ensemble of $M$ signals, each bandlimited to $W/2$ radians per second. The signals are {\em correlated}, i.e., $M$ signals can be well approximated by the linear combination of $R$ underlying signals. Therefore, we can write $M$ signals in  ensemble $\vm{X}_c(t)$ (on the left) as a tall matrix (a correlation structure) multiplied by an ensemble of $R$ underlying independent signals.}
  \label{fig:CM-Fig1}
\end{figure} 

\begin{figure}[htp]
  \begin{center}
	\includegraphics[height=2in]{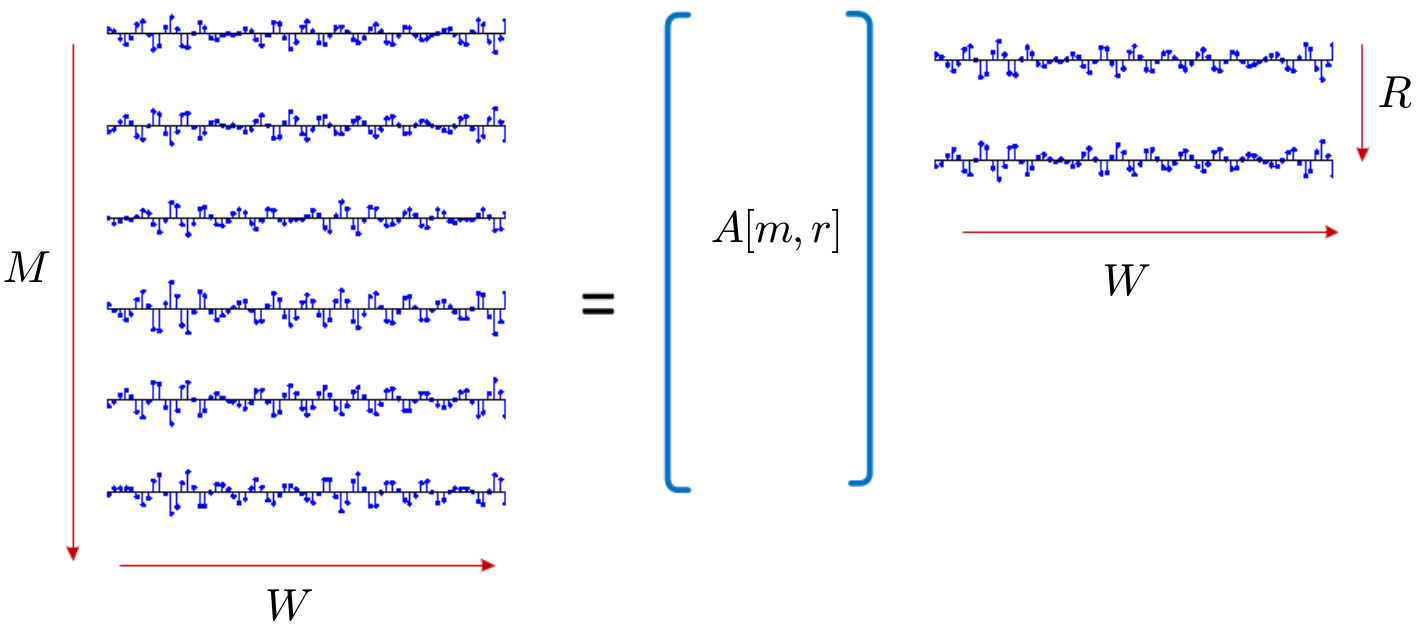}
    \end{center}
    \caption{\small\sl Samples $\vm{X}$ of ensemble $\vm{X}_c(t)$ inherit the low-rank property. Therefore, the problem of recovering $\vm{X_c(t)}$ from samples at a sub-Nyquist rate can be recast as a low-rank matrix recovery problem from partial-generalized measurements.}
  \label{fig:CM-Fig2}
\end{figure} 

This paper analyzes the two compressive multiplexing architectures illustrated in Figures~\ref{fig:CM-M-Mux} and \ref{fig:CM-FM-Mux}.  The first architecture, which we call M-Mux (for Modulated Multiplexing), can be broken into two parts. First, the $M$ input signals $\{x_m(t)\}_{ 1\leq m \leq M}$ are modulated against binary waveforms $\{d_m(t)\}_{ 1\leq m \leq M}$.  The minimum distance between polarity changes in the $d_m(t)$ is $1/\O$.  Second, the signals are added together and then sampled uniformly at rate $\O$ to produce measurements $y[n]$.  Theorem~\ref{thm:CM-exactrec-CM1} below shows that if the input ensemble can be written as a linear combination of $R$ latent signals (as in Figure~\ref{fig:CM-Fig1}),
\[
	\x_c(t) = \{x_m(t): x_m(t) = \sum_{ r = 1}^R A[m,r]s_r(t),~ 1\leq m \leq M\},
\]
and the energy in the signals is not too concentrated in a short interval of time, then they can be recovered when $\O\sim R(M+W)\log^3(MW)$.  When $R\ll M$, then this improves on the cumulative Nyquist rate of $MW$.  The second architecture, shown in Figure \ref{fig:CM-FM-Mux}, adds a linear time-invariant filter in front of the modulators whose purpose is to ensure that the signals are spread out in time --- we call this FM-Mux (Filtered and Modulated Multiplexing).  If the impulse responses of these filters are long and diverse, then the signal ensemble can be recovered when $\O\sim R(M+W)\log^5(MW)$ regardless of its structure in time; this is codified in Theorem~\ref{thm:CM-exactstablerec-CM2}.

\begin{figure}[htp]
  \begin{center}
    \includegraphics[trim = 2cm 11cm 2cm 0cm, scale = 0.6]{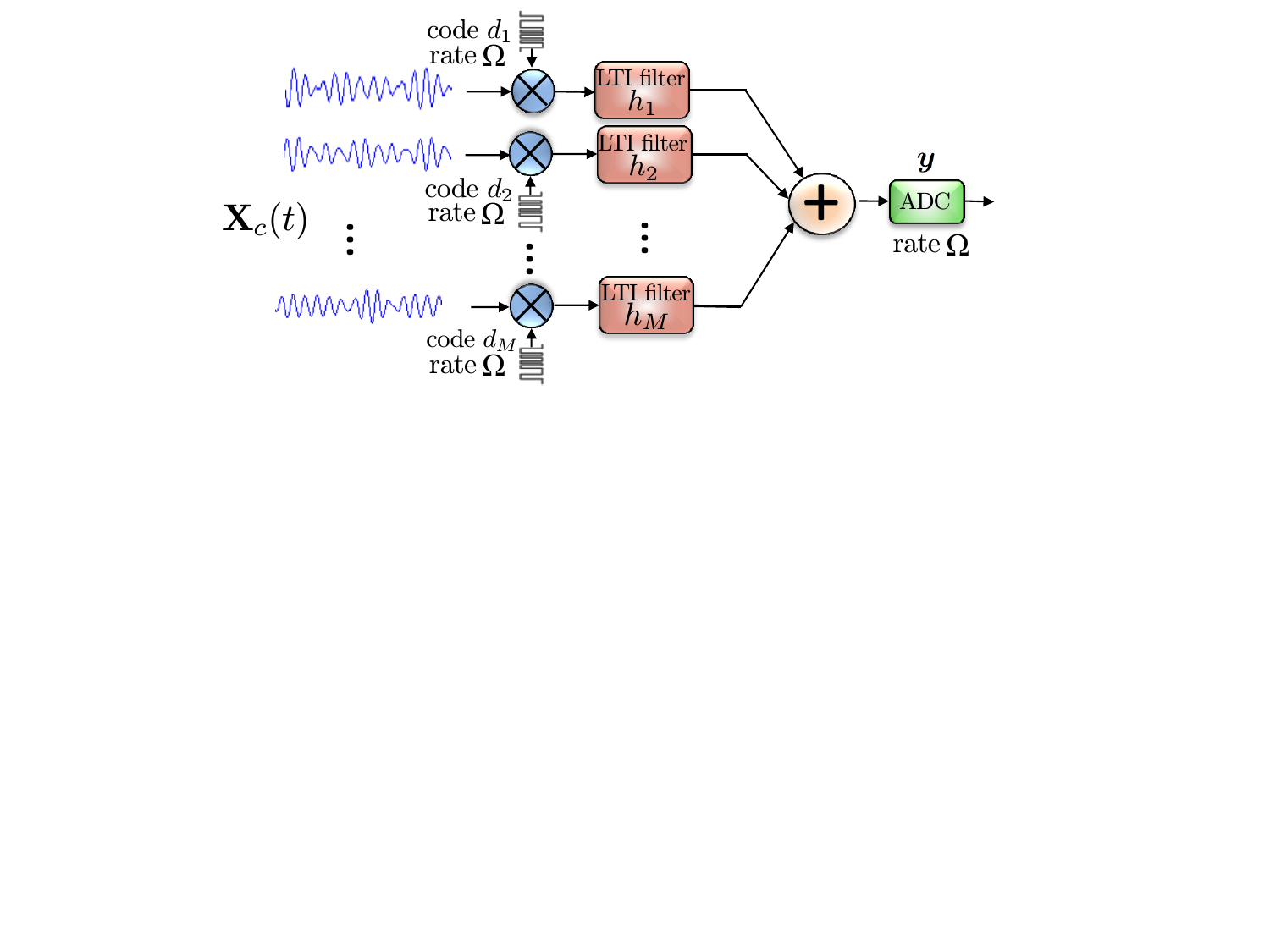}
    \end{center}
    \caption{\small\sl The FM-Mux for the efficient acquisition of correlated signals. Each of the input signal $\{x_m(t)\}_{ 1\leq m \leq M}$ is modulated separately with $\pm 1$-binary waveform $\{d_m(t)\}_{ 1 \leq m \leq M}$ alternating at rate $\O$. Afterward, the signals are convolved with diverse waveforms using random LTI filters in each channel. The resultant signals are then combined and sampled at a rate $\O$ using a single ADC.}
  \label{fig:CM-FM-Mux}
\end{figure}

We will use different mathematical tools to analyze these two multiplexing architectures.  The arguments for the FM-Mux (Figure~\ref{fig:CM-FM-Mux}) are more straightforward, and this architecture is more powerful in that it is universal (i.e.\ it is effective for any type of correlation structure and signal energy distribution).  However, it is probably the case that the M-Mux (Figure~\ref{fig:CM-M-Mux}) is more practical; in fact, this type of multichannel random modulator has been implemented previously for applications in radar signal processing and communications \cite{laska07th,yoo12a100,yoo12co,mishali11xasub,murray11de}.

The paper is organized as follows.  In the remainder of this section, we present some applications and the related work. Section \ref{sec:CM-MainResults} illustrates main results and sampling theorems for each of the multiplexing architecture. Section \ref{sec:CM-Exps} contains some illustrative numerical simulations.  Sections~\ref{sec:CM-Theory1}, \ref{sec:CM-Theory2}, and \ref{sec:CM-Theory3} provide the proofs of the sampling theorems.

%---------------------------------------------------------------------------------------
\subsection{Notation}
Unless specified otherwise, we use uppercase bold, lowercase bold, and not bold letters for matrices, vectors, and scalars, respectively. For example, $\vm{X}$ denotes a matrix, $\vm{x}$ represents a vector, and $x$ refers to a scalar. Calligraphic letters such as $\mathcal{A}$ specify linear operators. The letter $c$ refers to a constant number, which may
not refer to the same number every time it is used. The notations $\|\cdot\|$, $\|\cdot\|_*$, and $\|\cdot\|_{\F}$ denote the
operator, nuclear, and Frobenius norms of the matrices,
respectively. Furthermore, we will use $\|\cdot\|_2$, and $\|\cdot\|_1$
to represent the vector $\ell_2$, and $\ell_1$ norms.

%-------------------------------------------------------------------------------------
\subsection{Example application: Micro-sensor arrays}
\label{sec:CM-app}

In many applications in array processing, wavefronts incident on a large number of closely located antenna arrays generate signals that are highly correlated. This is especially true for micro-sensor arrays found, for example, in modern on-chip radars, tactile sensors in robotics, and microelectrode arrays (MEAs) used to study neural activity. 
In several of these array processing applications, we want to estimate signal parameters, such as angle of arrival, and frequency offsets. The first step towards achieving this is to estimate the covariance matrix of the input signal ensemble, and then use this to further estimate particular parameters (one example of this is the MUSIC algorithm \cite{schmidt86mu} for multiple emitter direction of arrival estimation).
The rank of the covariance matrix of a correlated signal ensemble
composed of $R$ latent independent signals is always $R$. 
In an on-chip radar, and other micro-sensor array applications, where limiting the number of samples might help meet design constraints (by reducing power, etc), compressive multiplexers can be used to estimate the covariance matrix from a smaller number total of samples on a single line than sampling each signal at the Nyquist rate directly. 

As multiplexing is a particular challenge in several biosensing applications, we will briefly discuss some motivating details of one such application, where the task is to monitor neural activity in brain tissues.

Neuronal recordings are used to study how different stimuli are encoded and processed by the firing of neurons.  The recordings are made by inserting an array of electrodes into the brain of an animal, and measuring the electrical activity.  Figure~\ref{fig:CM-NeuralApp} illustrates a typical geometry for such a device, and contains plots of a recording in an actual experiment performed as part of an effort to understand neuronal activity resulting from certain types of a visual stimuli.  This particular experiment\footnote{The data used in this figure comes from crcns.org, an open database for brain experimental data; the particular dataset can be found at \cite{dataset}.} used a microelectrode array containing 54 recording sites, and the plots in Figure~\ref{fig:CM-NeuralApp}(b) make it clear that subsets of the signals are highly correlated.

\begin{figure*}[htp]
  \begin{center}
    \includegraphics[scale =0.5 ]{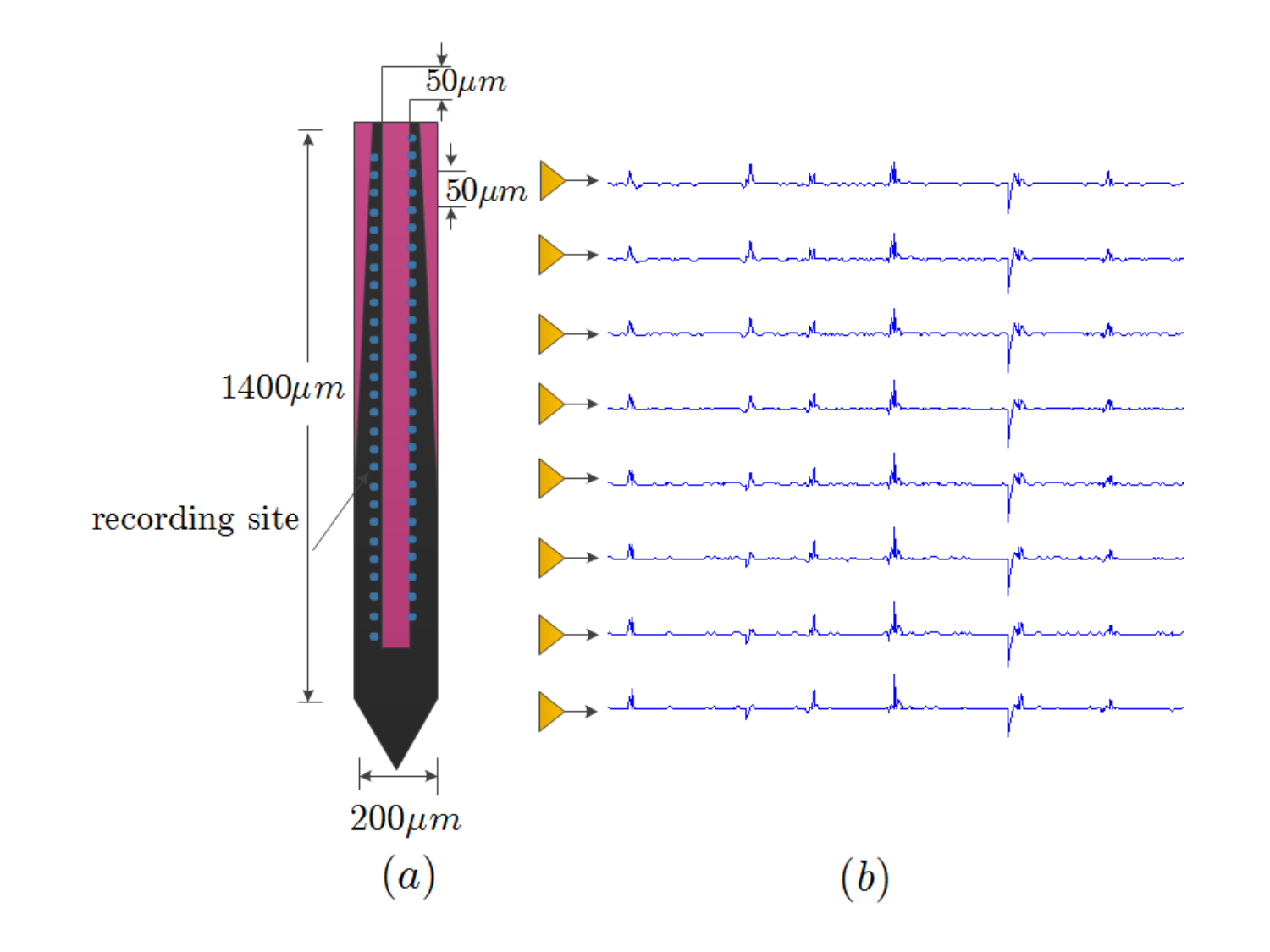}
    \end{center}
    \caption{\small\sl Application in neuronal recordings from brain tissues. (a) A microelectrode array with fifty-four recording sites, shown as blue dots, arranged in two columns 50$\mu m$ apart. Arrays with dense recording sites provide detailed field recordings and span roughly 1$mm$ of the brain tissue\cite{blanche2005polytrodes}. (b) The signals recorded by sensors in a real experiment. The data is taken from \cite{dataset}.}
  \label{fig:CM-NeuralApp}
\end{figure*}

In general, high density MEAs containing tens of thousands of recording sites; see, for example, \cite{frey07ce,imfeld08la,haas08pr,gray04di}, are used to record measurements at a high spatial and temporal resolutions in various biosensing applications. The thousands of signals recorded are multiplexed, continuously sampled by ADCs, and streamed to a hard disk at a high quantization resolution. This process generates massive amounts; on the orders of several gigabits per second (Gbps), of data. In particular, \cite{imfeld08la} describes a data acquisition platform for a microelectrode array containing 4096 recording sites. The signals are multiplexed onto fewer channels and then acquired using ADCs. The rate at which ADCs operate is determined by the acquisition requirement of 12-bit quantization resolution with a sampling rate of 20,000 samples per second for each of the 4096 recorded signals. This generates data roughly at 0.5 Gbps. It is clear that the sampling burden on the ADCs increases with increasing density of the recording sites on the MEAs, and so does the amount of the data generated; especially, for experiments lasting over many hours.   This calls for more proactive acquisition strategies for data acquisition, transfer, and management. The proposed compressive multiplexers use the correlation in the signal ensemble to acquire the signals with fewer samples to effectively use the sampling resources, and to minimize the amount of data generated over the course of an experiment. 

Another design consideration in MEAs  is that the number of electrodes on an array is limited by the number of conductors, carrying the signal from each electrode, that can pass through its shank. If we can perform an on-chip multiplexing then the signals can be combined before passing through the shank. This reduces the number of conductors, which may assist in increasing the density of recording sites for a given thickness.  Since the multiplexing architecture uses simple modulators, it may be possible to built these devices on chip. Additionally, the reduction in the sampling rate reduces the power dissipation of the ADC, which is an important factor in applications in biosensing.

%---------------------------------------------------------------------------------------
\subsection{Related work}

The modulated multiplexer (M-Mux) has been proposed previously in the literature \cite{slavinsky11co} for the compressive acquisition of multiple spectrally sparse signals.  Using the notation of this paper, the main results suggest that if the Fourier spectrum of the input signals can be approximated by  active frequency components $S \ll MW$, then \cite{romberg2010sparse} shows that for the successful reconstruction of the signal ensemble, the ADC is required to operate at rate $\Omega \approx S\log^q MW$, where $q>1$ is a small constant. A simple implementation of the M-Mux using a passive averager is also discussed in \cite{slavinsky11co}. 

Compressive sampling of spectrally sparse signals using random modulators has also been explored previously in the literature  \cite{tropp2010beyond,mishali2009blind} and have been implemented in hardware for multiple applications \cite{laska07th,yoo12a100,yoo12co,mishali11xasub,murray11de} --- the existence of these prototypes was one of the primary factors that lead us to consider the M-Mux.  Instead of considering the acquisition of a single sparse signal, this paper considers the joint acquisition of an ensemble of signals.  Structure is imposed on this ensemble not by imposing structure on each of the signals individually, but rather on the {\em relationships} between the signals.  This requires a completely different recovery technique, and a new set of analytical tools.

It will be shown in detail in Section~\ref{sec:sampling-thm-m-mux} that the $n$th sample $y[n]$ taken using the ADC of the M-Mux can be written as the trace inner product of an unknown  rank-$R$ matrix $\vm{C}_0$ against a rank-1 measurement matrix $\vm{A}_n$, i.e., $y[n] = \tr{\vm{C}_0\vm{A}_n^*}$, where $\vm{A}_n$ is formed by the outer product of a random vector with a Fourier vector --- Theorem \ref{thm:CM-exactrec-CM1} proves that the low-rank matrices can be successfully recovered using such rank-1 measurement matrices. Similar results showing the recovery of low-rank matrices using rank-1 measurement matrices have been the subject of some interesting recent literature; for example, \cite{candes2013phaselift,candes2012solving,demanet2014stable}. In these articles, the measurement matrices are rank-1 but are formed by the outer product of a random vector with itself. The measurement matrices in this paper also differ from the measurement model in \cite{recht10gu}, where each of the measurement matrix is an i.i.d. Gaussian random matrix and it is shown that RIP based \textit{stronger} recovery results are possible. It is also instructive to compare the results in this paper with the results in \cite{gross11re} that state that it is possible to recover a low-rank matrix by observing its random samples in an incoherent orthonormal basis $\{\vm{A}_n\}$. The measurement matrices in our case do not form an orthonormal basis and owing to their special structure, we only require incoherence on one set of the singular vectors of the unknown low-rank matrix $\vm{C}_0$.

As will be shown in Section~\ref{sec:CM-M-Mux}, the samples taken by the ADC in Figure~\ref{fig:CM-M-Mux} can be mathematically modeled as a multi-Toeplitz matrix acting on a vectorized version of the collection of Fourier coefficients for the signals in the ensemble.  For ensembles with just one independent component ($R=1$), the analysis is a special case of the main results in the recent paper \cite{ahmed2012blind}.  That reference is a study of a very different application, namely, blind deconvolution of two unknown signals.  The mathematics presented here extends the analysis of that paper to the recovery of rank $R$ matrices.

One of the compressive multiplexing architectures we consider in this paper involves pre-filtering the signals using filters with long, diverse impulse responses (which we generate randomly).  Previous work has shown that a low-rate sampling preceded by a convolution with a random waveform is an effective strategy for compressive sampling acquisition of sparse signals \cite{tropp2006random, romberg2009compressive, haupt10to,rauhut12re}. Results in these references show that a signal with $S$ active components in a fixed basis can be acquired using a random filter plus an ADC operating at a rate that scales linearly in $S$ and logarithmically in ambient dimension $W$. 

%---------------------------------------------------------------------------------------
\section{Main results: Sampling theorems for compressive multiplexers}
\label{sec:CM-MainResults} 

In this section, we present the mathematical models for the signal ensemble and for the samples taken by each of the proposed compressive multiplexer architectures.  The signal ensemble is characterized by a $M\times W$ low-rank matrix, while the mapping from the ensemble to the sample at the output of the ADC is a linear operator acting on this matrix.  With the model in place, we state our sampling theorems in Sections~\ref{sec:sampling-thm-m-mux} and \ref{sec:sampling-thm-fm-mux}.

%--------------------------------------------------------------------------------------------------
\subsection{Signal model}

We will use $\vm{X_c}(t)$ to denote a signal ensemble of interest and $x_1(t),\ldots,x_M(t)$ to denote the individual signals within that ensemble.  Conceptually, we may think of $\vm{X_c}(t)$ as a ``matrix'' with finite $M$ number of rows, but each row contains a bandlimited signal.  Our underlying assumption is that the signals in the ensemble are {\em correlated} in that 
\begin{equation}
	\label{eq:CM-lowrankensemble}
	\vm{X}_c(t) \approx \vm{A}\vm{S}_c(t),
\end{equation}
where $\vm{S}_c(t)$ is a smaller signal ensemble with $R$ rows and $\vm{A}$ is a $M\times R$ matrix with entries $A[m,r]$.  We will use the convention that fixed matrices operating to the left of the signal ensembles simply ``mix'' the signals point-by-point, and so \eqref{eq:CM-lowrankensemble} is equivalent to
\[
	x_m(t) \approx \sum_{r=1}^R A[m,r]s_r(t).
\]

The only structure we will impose on individual signals is that they are real-valued, bandlimited, and periodic.  This provides us with a natural way to discretize the problem, as each signal lives in a finite-dimensional linear subspace.  The periodicity assumption is made mostly to keep the mathematics clean; in Section~\ref{sec:CM-nonperiodic}, we discuss how our results can be adapted to more realistic signal models in which non-periodic signals are windowed into overlapping sections and reconstructed jointly. Each bandlimited periodic signal in the ensemble can be written as a Fourier series
\[
	x_m(t) = \sum_{\omega =-B}^B\alpha_m[\omega]\, \er^{\j2\pi \omega t},
\]
where $\alpha_m[\omega]$ are complex but have symmetry $\alpha_m[-\omega]=\alpha_m[\omega]^*$ to ensure that $x_m(t)$ is real.  The signals are equally well represented by the $W=2B+1$ Fourier coefficients $\alpha_m$, or by $W$ equally spaced time-domain samples.

The modulation codes $d_m(t)$ will in general be changing polarity at a rate $\Omega > W$. We can generate an  $M \times \O$ matrix $\vm{X}_0$ of samples of the signals at this rate by taking
\begin{equation}
	\label{eq:CM-C0-def}
	\vm{X}_0 = \vm{C}_0\tilde{\vm{F}},
\end{equation}
where $\tilde{\vm{F}}$ is a $W\times \O$ matrix formed by taking first $W$ rows of the normalized discrete Fourier matrix $\vm{F}$ with entries
\begin{equation}
	\label{eq:Fdef}
	F[\omega,n] = \frac{1}{\sqrt{\O}} \er^{-\j2\pi\omega n/\O},\quad 0\leq\omega,n\leq \O-1,
\end{equation}
and $\vm{C}_0$ is an $M\times W$ matrix whose rows contain Fourier series coefficients for the signals in $\vm{X}_c(t)$. 
\[
	C_0[m,\omega] = 
	\begin{cases} 
		\alpha_m[\omega] & \omega = 0,1,\ldots,(W-1)/2 \\
		\alpha_m[\omega-W]^* & \omega = (W+1)/2,\ldots,W-1
	\end{cases}.
\]
The matrix $\vm{F}$ is orthonormal, while $\vm{C}_0$ (and hence $\vm{X}_0$) inherits the correlation structure of the original ensemble.  Our efforts will be geared towards recovering the matrix $\vm{C}_0 \in \comps^{M \times W}$ which uniquely specifies the signal ensemble. 

We will consider both the case in which $\vm{C}_0$ is exactly rank $R$, and the case in which $\vm{C}_0$ is technically full rank  but can be closely approximated by a low-rank matrix (i.e., the spectrum of singular values decays rapidly). 

\subsection{M-Mux: Compressive multiplexing of time-dispersed correlated signals}
\label{sec:CM-M-Mux}

In this section, we develop the mathematical model for the samples taken by the ADC in the M-Mux, shown in Figure~\ref{fig:CM-M-Mux}.  The end result will be to write the samples as a discrete linear transformation of the discretized input signals.

The multiplexer contains $M$ input channels carrying signals $x_m(t)$ which it modulates against different binary $\pm 1$ waveforms $d_m(t)$.  The $d_m(t)$ have higher bandwidth than the input signals; the spacing between the possible transition points is $1/\Omega$, where $\Omega > W$.  Since sampling the signals commutes with their addition, we can equivalently add the rate $\Omega$ samples of modulator outputs $\{d_m(t)x_m(t)\}_{1\leq m \leq M}$ to produce the samples.  We can write the $\Omega$ samples $\vm{y}_m$ of $d_m(t)x_m(t)$ on $[0,1)$ as 
\[
	\vm{y}_m = \vm{D}_m\tilde{\vm{F}}^*\vm{c}_m,
\]
where $\vm{c}_m$ is the $W$-vector containing the Fourier coefficients of $x_m(t)$, $\tilde{\vm{F}}^*$ is the $\O\times W$ (oversampled) inverse Fourier matrix as in \eqref{eq:Fdef}, and $\vm{D}_m$ is an $\Omega\times\Omega$ diagonal matrix constructed from the $\O$ samples $\vm{d}^{(m)} = \{d_1[m],\ldots,d_\O[m]\}$ of $d_m(t)$. The ``tall'' Fourier matrix $\tilde{\vm{F}}^*$ is an interpolation matrix that produces samples of the signals at the same rate $\Omega$ as the switching times of the $d_m(t)$.

The modulation signals $d_m(t)$ are generated from random sign sequences, which means $\vm{D}_m$ is a random matrix of the following form:
\begin{equation}
\label{eq:CM-modulatorD}
	\vm{D}_m = 
	\begin{bmatrix}
		d_1[m] & & & \\
		& d_2[m] & & \\
		& & \ddots & \\
		& & & d_{\O}[m]
	\end{bmatrix}
	\quad 
	\text{where $d_{\o}[m]= \pm 1$ with probability $1/2$}, 
\end{equation}
and the $d_{\o}[m]$ are independent $\forall (\o,m) \in \{1,\ldots,\O\}\times\{1,\ldots,M\}$.  In the sequel, we use the superscript notation $\vm{d}^{(m)}$ to specify $[d_1[m],\ldots,d_\O[m]]^{\T}$, the collection of samples of the modulation waveforms across all channels at a fixed time; we use the subscript notation $\vm{d}_{\o}$ for $[d_\o[1],\ldots,d_\o[M]]^{\T}$, the collection of samples of a single modulation waveform over the entire time interval.  

Conceptually, the modulators are  embedding each of the $x_m(t)$ into different (but overlapping) subspaces of $\R^\O$ --- this is what allows us to ``untangle'' them after they have been added together.

The ADC takes $\Omega$ samples of $\sum_{m = 1}^M x_m(t)d_m(t)$ on $[0,1)$. We can write the vector of samples $\vm{y}$ as
\begin{align}
	\label{eq:CM-M-Mux-meas}
	\vm{y} &= \sum_{m = 1}^M \vm{D}_m\tilde{\f}^*\vm{c}_m = 
	[\vm{D}_1\tilde{\f}^*, \vm{D}_2\tilde{\f}^*, \cdots, \vm{D}_M\tilde{\f}^*]\cdot\mbox{vec}(\vm{C}_0^*)\notag\\
	&= \cA(\vm{C}_0),
\end{align}
where $\vm{C}_0$ is the $M \times W$  matrix with $\vm{c}_m^*$ as its rows and vec$(\cdot)$ takes a matrix and returns a vector obtained by stacking its columns.  In the last equality,  we combines all of these actions into a single linear operator $\cA: \comps^{M \times W} \rightarrow \reals^{\O}$ which takes as input the matrix of Fourier coefficients $\vm{C}_0$ of the input signals, and outputs the $\Omega$ samples.

Looking at the architecture in Figure~\ref{fig:CM-M-Mux}, we expect that the M-Mux will perform better for signals ensembles which are not too concentrated in time.  Although the fact that the mixers and the ADC are operating at a rate above $W$ means that we will get multiple ``looks'' at a signal no matter what, it also true that if all of the signals are concentrated in the same subinterval instead of being spread out in time, we are getting fewer effective samples to distinguish between them.  This intuition is supported by our theoretical analysis for the M-Mux.  As we will see later that the sampling performance of the M-Mux depends on a mild incoherence condition, which quantifies the dispersion of the input signal ensemble across time.

\subsection{FM-Mux: A universal compressive multiplexer for correlated signals}
\label{sec:CM-FM-Mux}

In this section, we present a modified version of the M-Mux which is {\em universal} in that it is effective no matter how the energy in the signals is dispersed in time, or how they are correlated.  The architecture, shown in Figure~\ref{fig:CM-FM-Mux}, adds a set of linear time-invariant (LTI) filters in between the modulators and the signal summation.  Their effect is to spread the signals out in time.  We call this filtered modulated multiplexer the FM-Mux.

The FM-Mux preprocesses the input signals as follows.  First, the signals are modulated against a $\pm 1$-binary waveform with switching rate $\Omega >W$; this disperses the frequency spectrum of the signals over a larger bandwidth roughly proportional to $\O$.  Second, the signals are convolved with impulse responses $h_m(t)$ that are long and diverse, diffusing the signal across time.  Finally, the signals are added together and sampled uniformly at rate $\Omega$. 

%\subsubsection{Analog Preprocessing using Modulators and Filters}
As before, the modulators in the FM-Mux take the input signals $x_1(t), \ldots, x_M(t)$ and multiply them with $d_1(t), \ldots, d_M(t)$, where the $d_m(t)$ have the same properties as the M-Mux described in the previous section. 
%Each of the $d_m(t)$ is an independent and random binary $\pm 1$ waveform that is constant over a time interval of length $1/\Omega$, where $\Omega >W$ and $W$ is the bandwidth of the signals. 
The filters in the $m$-th channel takes the modulated signals $x_m(t)d_m(t)$, which are bandlimited to $\Omega/2$, and convolves them with an impulse response $h_m(t)$ which we will specify.  We will assume that we have complete control over this impulse response, putting practical implementation issues aside.  We write the action of the LTI filter $h_m(t)$ as an $\Omega \times \Omega$ circular matrix $\vm{H}_m$ (the first row of $\vm{H}$ consists of samples $\vm{h}_m$ of $h_m(t)$) operating on the Nyquist rate samples $\vm{D}_m\tilde{\vm{F}}^*\vm{c}_m$ in $[0,1)$  of $x_m(t)d_m(t)$.  The circulant matrix $\vm{H}_m$ is diagonalized by the discrete Fourier transform:
\begin{equation*}
\vm{H}_m = \vm{F}^*\hat{\vm{H}}_m\vm{F},
\end{equation*}
where $\hat{\vm{H}}_m$ is a diagonal matrix whose entries are $\hat{\vm{h}}_m = \sqrt{\Omega}\vm{F}\vm{h}_m$.  The vector $\hat{\vm{h}}_m$ is a scaled version of the non-zero Fourier series coefficients of $h_m(t)$.

To generate the impulse response, we will use a random unit-magnitude sequence in the Fourier domain\cite{tropp2006random,romberg2009compressive}.  In particular, we will take 
\begin{equation*}
	\hat{\vm{H}}_m =
	\begin{bmatrix}
		\hat{h}_m(0) & & & \\
		& \hat{h}_m(1) & & \\
		& & \ddots & \\
		& & & \hat{h}_m(\Omega-1)
	\end{bmatrix},
\end{equation*}
where
\begin{equation*}
	\hat{h}_m(\omega) = 
	\begin{cases}
		\pm 1,\text{with prob.\ $1/2$}, & \omega = 0 \\
		\er^{j\theta_\omega}, ~\text{where} ~\theta_\omega\sim\mathrm{Uniform}([0,2\pi]),  & 1\leq \omega \leq (\Omega-1)/2 \\ 
		\hat{h}_m(\Omega-\omega)^*, & (\Omega+1)/2\leq\omega\leq \Omega-1
	\end{cases}.
\end{equation*}
These symmetry constraints are imposed so that $\vm{h}_m$ (and hence, $h_m(t)$) is real-valued.  Conceptually, convolution with $h_m(t)$ disperses a signal over time while maintaining fixed energy (note that $\vm{H}_m$ is an orthonormal matrix).

Given the discussion above, the Nyquist samples of $(x_m(t)d_m(t))*h_m(t)$ are given by the $\O$-vector $\vm{H}_m\vm{D}_m \tilde{\vm{F}}^*\vm{c}_m$, and the samples $\vm{y}$ in $[0,1)$ of the signal $y(t) = \sum_{m = 1}^M (x_m(t)d_m(t))*h_m(t)$ are
\begin{align}
	\vm{y} &= \sum_{m = 1}^M \vm{H}_m\vm{D}_m\tilde{\vm{F}}^*\vm{c}_m\notag \\
       	&=  [\vm{H}_1\vm{D}_1\tilde{\vm{F}}^*, \vm{H}_2\vm{D}_2\tilde{\vm{F}}^*, \cdots ,
 	\vm{H}_M\vm{D}_M\tilde{\vm{F}}^*]\cdot\mbox{vec}(\vm{C}_0^*)\notag\\	
       & = \cB(\vm{C}_0)
	\label{eq:CM-PhiD}
\end{align} 
where we have used $\cB: \comps^{M \times W} \rightarrow \reals^{\O}$ to denote the linear transformation encapsulating all of the steps above. The linear operator $\cB$ is a random block-circulant matrix with columns modulated by random signs, i.e., the randomness appears in a  structured form.

The positions of the modulators and filters can be swapped, as illustrated in Figure~\ref{fig:CM-FM-Mux2}.  In this case, it will be sufficient to use filters of bandwidth $W$ rather than the bandwidth of $\O$ bandwidth used in Figure~\ref{fig:CM-FM-Mux}. The theoretical analysis for this swapped architecture is very similar to the FM-Mux in Figure \ref{fig:CM-FM-Mux}; for simplicity we only state the formal result for the first architecture, but we discuss how the analysis of the second architecture is related at the end of Section \ref{sec:CM-Theory3}.

\begin{figure}[htp]
  \begin{center}
    \includegraphics[trim = 2cm 11cm 2cm 0cm, scale = 0.6]{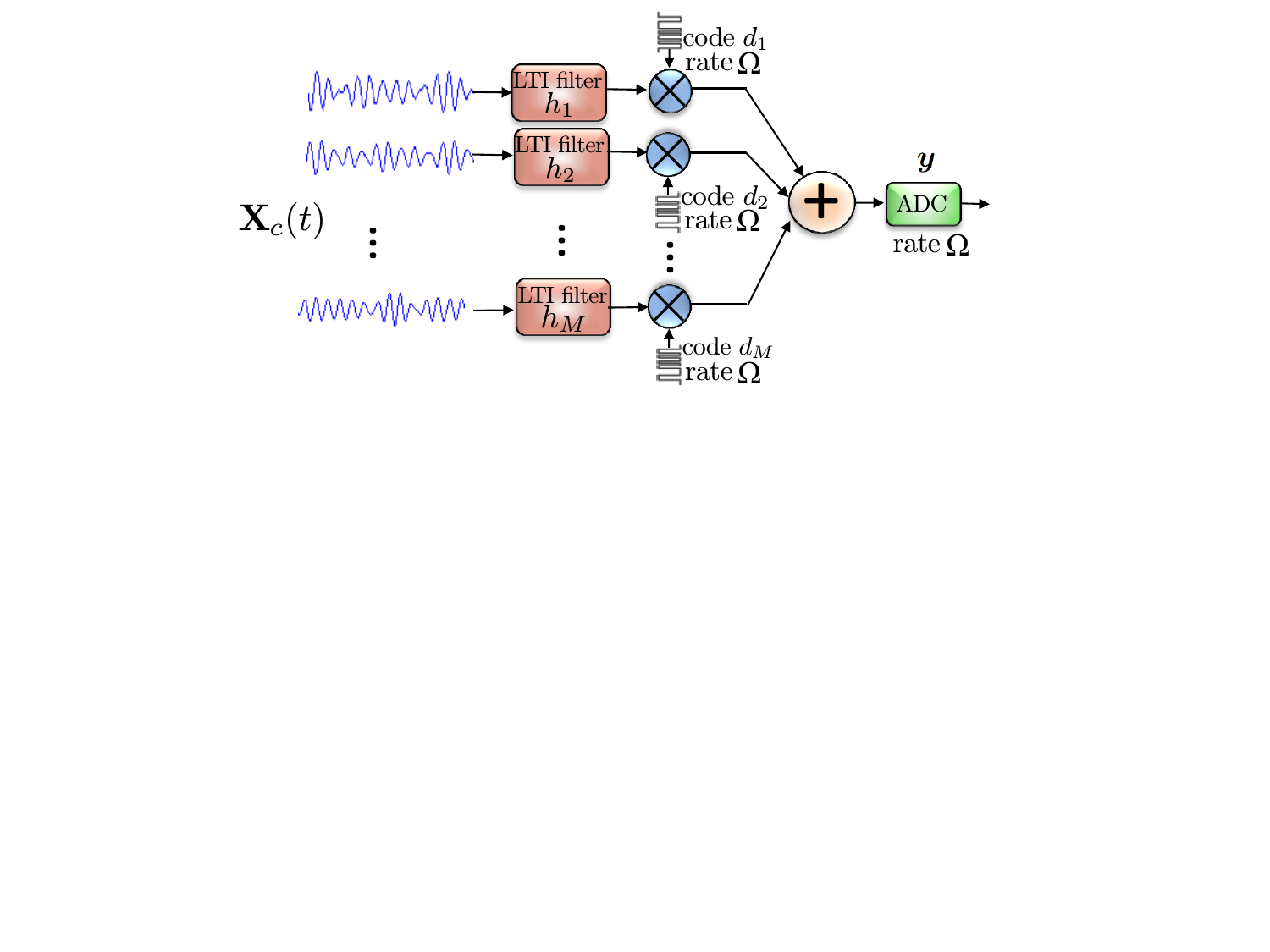}
    \end{center}
    \caption{\small\sl An equivalent FM-Mux obtained by reversing the order of filters and modulators. The modulators operates exactly as before, however, the random filters operate in a bandwidth $W$ instead of operating in a larger bandwidth $\O$ as in the previous FM-Mux architecture.}
  \label{fig:CM-FM-Mux2}
\end{figure}
%---

The presence of the filters $h_m(t)$ with random impulse responses makes the FM-Mux slightly less practical than the M-Mux.  However, its universality makes it a more powerful architecture, and as we will see in Sections~\ref{sec:sampling-thm-fm-mux} and \ref{sec:CM-Theory3} below, it is much easier to analyze mathematically.
 
\subsection{Methodology for signal reconstruction}

The samples $\vm{y}$ taken by the ADC in the M-Mux \eqref{eq:CM-M-Mux-meas} and in the FM-Mux \eqref{eq:CM-PhiD} are different linear transformations of the low-rank matrix $\vm{C}_0$ which we denote by $\cA$ and $\cB$, respectively.  The discussion in this section applies equally to both architectures, so we will use $\mathcal{T}$ to denote a generic linear measurement operator from $\comps^{M\times W}$ to $\reals^{\O}$.  We are given measurements 
\begin{equation}
	\label{eq:CM-M-Mux-meas1}
	\vm{y} = \mathcal{T}(\vm{C}_0),
\end{equation}
from which we wish to recover the original signal ensemble $\vm{C}_0$.

It is instructive to first consider the case when the correlation structure $\vm{A}$ in \eqref{eq:CM-lowrankensemble} is known. The matrix $\vm{C}_0$ in \eqref{eq:CM-C0-def} inherits the low-rank structure of $\x_0$, and can be decomposed as 
\[
	\vm{C}_0 = \vm{A}\vm{C}_s,
\]
where $\vm{C}_s \in \comps^{R \times W}$ is a coefficient matrix that contains the Fourier coefficients of the underlying signals $\{s_r(t)\}_{1 \leq r \leq R}$ as its columns. Define an operator $\mathcal{T}_{\vm{A}}:\comps^{W \times R} \rightarrow \reals^{\O}$ obtained by absorbing the known correlation structure $\vm{A}$ into the measurement process,  
\begin{align*}
\mathcal{T}_{\vm{A}} &= \mathcal{T}\circ\vec{\vm{A}}\\
& = \vm{T}\vec{\vm{A}}\\
&= \vm{T}\begin{bmatrix}
		 A[1,1]\vm{I} & A[1,2]\vm{I} & \hdots & A[1,R]\vm{I} \\
		 A[2,1]\vm{I} & A[2,2]\vm{I} & \hdots & A[2,R]\vm{I} \\
		\vdots & \vdots & \ddots & \vdots \\
		 A[M,1]\vm{I} & A[M,2]\vm{I} & \hdots & A[M,R]\vm{I}
	\end{bmatrix}, 
\end{align*}
where $\mathbf{T}$ is the $\O \times MW$ matrix representation of linear operator $\mathcal{T}$, $\vec{\vm{A}}$ is the $MW \times RW$ matrix on the right above, and $\vm{I}$ is the $W \times W$ identity matrix.  With the measured samples now written as
%This means that we can write the measurements $\vm{y}$ in \eqref{eq:CM-M-Mux-meas1} as 
\begin{align*}
	\vm{y} &= \mathcal{T}_{\vm{A}}(\vm{C}_s),
\end{align*}
and given that we are not making any structural assumptions about the $s_r(t)$, we can search for a coefficient matrix that is consistent with these samples by solving the least-squares program
\begin{align}
	\label{eq:CM-LeastSquares}
	\min_{\vm{C} \in \comps^{R \times W}} \quad \|\vm{y} - \mathcal{T}_{\vm{A}}(\vm{C})\|_2^2,
\end{align} 
the solution to which is given 
\[
	\tilde{\vm{C}}_s = (\mathcal{T}_{\vm{A}}^*\mathcal{T}_{\vm{A}})^{-1}\mathcal{T}_{\vm{A}}^*(\vm{y}).
\]
An argument similar to the proof of, for example, Lemma \ref{lem:CM-injectivity}, involving matrix Chenoff bounds can be used to show that $(\mathcal{T}_{\vm{A}}^*\mathcal{T}_{\vm{A}})^{-1}$ is well-conditioned with exceedingly high probability when the sampling rate $\O$ obeys
\begin{equation}
	\label{eq:CM-optimal-rate}
	\O \gtrsim c\beta R(W+M)\log^2(MW).
\end{equation}
Since the focus of this paper is on unknown correlation structure, we will not make this conditioning argument explicit.  
%The sampling rate $\O$ in \eqref{eq:CM-optimal-rate} is optimal to within a constant and log factors. 
The estimate $\tilde{\vm{C}}$ of the unknown is given by $\vm{C}_0$ is then $\tilde{\vm{C}} = \vm{A}\tilde{\vm{C}}_s$. 

We are primarily interested in the case where the correlation structure $\vm{A}$ is unknown.  In this case, we would require on the order of $MW$ samples to recover the ensemble using least-squares.  But by explicitly taking advantage of the fact that $\rank{\vm{C}_0}$ is low rank in the recovery, we can recover the ensemble from a number of samples comparable to \eqref{eq:CM-optimal-rate} even when $\vm{A}$ is unknown.  Given $\vm{y}$, we solve for $\vm{C}_0$ using the nuclear-norm minimization program:
\begin{align}
	\label{eq:CM-nuclearnorm_min}
	&\min \quad \|\vm{C}\|_*\\
	&\mbox{subject to} \quad \vm{y} = \mathcal{T}(\vm{C}),\notag
\end{align}
where $\left\|\vm{C}\right\|_*$ is the nuclear norm; the sum of the singular values of $\vm{C}$. Alternatively, when the measurements are contaminated by noise,
\[
	\vm{y} = \mathcal{T}(\vm{C}_0)+\vm{\xi},
\]
we solve the relaxed program
\begin{align}
	\label{eq:CM-nuclearnorm_minnoisy}
	&\min \quad \|\vm{C}\|_*\\
	&\mbox{subject to} \quad \|\vm{y} - \mathcal{T}(\vm{C})\|_2 \leq \eta.\notag
\end{align}
These programs can be solved efficiently for matrices with $\sim 10^6$ entries using any one of a number of existing software packages \cite{becker2012tfocs,becker2010templates,MINFUNC,recht11pa,Lee10}. Further research on algorithms to minimze the nuclear norm efficiently and to make the real-time reconstruction of wideband signals possible at a resonable computing cost will be an important challenge in the future research in this direction.

The number of degrees of freedom in the unknown-coefficient matrix $\vm{C}_0$ is approximately $R(W+M)$.  It is known that if $\mathcal{A}$ is a random projection, then we can obtain a stable recovery of matrix $\vm{C}_0$ in noise when the number of measurements $\Omega$ exceeds $cR(W+M)$ for a fixed constant $c$ \cite{recht10gu,candes11ti}. In addition, it is also known that if we directly observe a randomly selected subset of the entries of low-rank matrix $\vm{C}_0$ at random, then we can recover $\vm{C}_0$ exactly when the number of measurements roughly exceed $c\mu^2_0R(W+M)\log W$, where $\mu_0^2$ is the coherence of matrix $\vm{C}_0$; for details, see \cite{candes09ex,gross11re,recht11si}. In contrast, the measurements in \eqref{eq:CM-M-Mux-meas} and in \eqref{eq:CM-PhiD} are obtained as a result of structured-random operations. There are no matrix recovery results from such specialized linear measurements. This paper develops low-rank matrix recovery results for such structured-random measurement operations.

\subsection{Sampling Theorems for the M-Mux}
\label{sec:sampling-thm-m-mux}

Each entry $y[\o]$ of the measurement vector $\vm{y}$ in \eqref{eq:CM-M-Mux-meas} can be written as a trace inner product against a different $M\times W$ matrix $\A_\o$:
\begin{align*}
	y[\o] &= \<  \vm{C}_0,\A_\o\> = \tr{\vm{C}_0\A_\o^*} , \quad \o = 1, \ldots, \Omega,
\end{align*}
where
\begin{equation}
	\label{eq:CM-Aomega-def}
	\A_\o = \d_\o\ff_\o^*,
\end{equation}
is the rank-1 matrix formed by the outer product of $\vm{d}_\o =  [d_\o[1], \ldots, d_\o[M]]^\T$ (first defined in \eqref{eq:CM-modulatorD}), and the columns $\ff_\o$ of the partial Fourier matrix $\tilde{\vm{F}}$.  Let 
\begin{equation*}
	\vm{C}_0 = \vm{U}\vm{\Sigma}\vm{V}^*,
\end{equation*}
be the SVD of the rank-$R$ coefficient matrix $\vm{C}_0$, and so $\vm{U}:M\times R$ and $\vm{V}: W \times R$ have orthonormal columns, and $\vm{\Sigma}: R \times R$ is diagonal.  We quantify the signal dispersion across time using the coherence parameter
\begin{equation}
	\label{eq:CM-coherence}
	\mu^2(V) := \frac{\Omega}{R}\max_{1 \leq \o \leq \Omega}\|\vm{V}^*\vm{f}_\o\|_2^2.
\end{equation} 
A lower bound for $\mu^2(V)$ follows from summing both sides of \eqref{eq:CM-coherence} over $\o$, 
\begin{equation*} 
	\sum_{\o = 1}^\O \mu^2(V) ~\geq~ \frac{\Omega}{R} \sum_{\o = 1}^\O \|\vm{V}^*\vm{f}_\o\|_2^2 
	~=~ \frac{\Omega}{R} \|\v^*\|_{\F}^2,
\end{equation*}
and so $\mu^2(V) \geq 1$. The coherence $\mu^2(V)$ achieves this lower bound when  $\|\v^*\ff_\o\|_2^2 = \frac{R}{\O}$ for each $\o \in \{1,\ldots, \O\}$, meaning that the $\O$-point inverse Fourier transforms of the columns of $\v$ are flat.  In other words, the signals are well dispersed across time.  An upper bound for $\mu^2(V)$ is given by 
\begin{equation*} 
	\mu^2(V) \leq \frac{\Omega}{R}\max_{1 \leq \o \leq \Omega} \|\v^*\|^2\|\ff_\o\|_2^2 \leq \frac{W}{R}.
\end{equation*}
The coherence achieves this upper bound for signal ensembles that are as concentrated in time as possible (e.g.\ sinc functions).

The following theorem guarantees the exact recovery of the ensemble $\x_c(t)$ at a sub-Nyquist sampling rate, when $\x_c(t)$, and hence $
\vm{C}_0$, is \textit{exactly} rank-$R$, that is, instead of \eqref{eq:CM-lowrankensemble}, we have $\x_c(t) = \vm{A}\vm{S}_c(t)$.
\begin{thm}
	\label{thm:CM-exactrec-CM1}
	Let $\vm{C}_0 \in \comps^{M \times W}$ be a matrix of rank $R$ defined in \eqref{eq:CM-C0-def} with coherence $\mu^2(V) \leq \mu^2_0$. Suppose $\Omega$ measurements $\vm{y}$ of $\vm{C}_0$ are taken using the M-Mux as in \eqref{eq:CM-M-Mux-meas1}. If
	\[
		\Omega \geq c\beta \left(\mu^2_0 M+W\right)R \log^3(MW),
	\]
	for some constant $\beta > 1$, then the minimizer of \eqref{eq:CM-nuclearnorm_min} is unique and equal to $\vm{C}_0$ with probability at least $1- O(MW)^{1-\beta}$. 
\end{thm} 
The sampling theorem above indicates that the time dispersed correlated signals ($\mu_0^2 \approx O(1)$) can be acquired at a sampling rate close (to within a $\log$ factor) to the optimal sampling rate $R(W+M)$. This is a significant improvement over the cumulative Nyquist rate $MW$ especially when $R \ll \min(M,W)$. The above result is also important as it is a low-rank matrix recovery result from a linear transformation $\cA$, which can be applied more efficiently compared to the dense, \textit{completely} random linear operators such as i.i.d. Gaussian linear operators. 

The recovery can be made stable in the presence of noise.  Now say we observe
\begin{equation}
	\label{eq:CM-noisy-meas}
	\vm{y} = \cA(\vm{C}_0) + \vm{\xi} 
\end{equation}
where $\vm{\xi}\in\R^\O$ is a noise vector, and $\vm{C}_0$ is \textit{exactly} rank-$R$.  One option is to solve the relaxed nuclear norm problem in  \eqref{eq:CM-nuclearnorm_minnoisy}, and indeed the numerical experiments shown in Section~\ref{sec:CM-Exps} show that this seems to recover the ensemble effectively.  Unfortunately, our efforts to analyze this program have resulted in only very weak stability results.  In this paper, we will consider the simpler recovery strategy from \cite{koltchinskii10nu}, which sets
\begin{align}
	\label{eq:CM-KLT-est}
	\tilde{\vm{C}} &= \argmin_{\vm{C}}\left[\|\vm{C}\|_{\F}^2 - 
					2\<\vm{y},\cA(\vm{C})\>+\lambda\|\vm{C}\|_*\right],
\end{align}
for a fixed value of the regularization parameter $\lambda >0$.  The program above, which we will call the {\em KLT estimator}, does not perform empirically as well as \eqref{eq:CM-nuclearnorm_minnoisy}, but its analysis proves far less elusive; in the end, we will show through Theorem~\ref{thm:CM-stablerec-CM1} below that near-optimal recovery from noisy measurements is possible with a nuclear norm penalized estimator.  The essential difference between the KLT estimator and \eqref{eq:CM-nuclearnorm_minnoisy} is that $\cA$ is explicitly treated as being random in the formulation.  The solution to \eqref{eq:CM-KLT-est} is found by soft thresholding the singular values of $\cA^*(\vm{y})$:
\[
	\tilde{\vm{C}} = \sum_i (\sigma_i(\cA^*(\vm{y}))-\lambda/2)_+ \uu_i(\cA^*(\vm{y}))\vv_i(\cA^*(\vm{y})),
\]
where $x_+ = \max(x,0)$, the vectors $\uu_i(\cA^*(\vm{y}))$, and $\vv_i(\cA^*(\vm{y}))$ are the left and right singular vectors of $\cA^*(\vm{y})$, respectively, and the  $\sigma_i(\cA^*(\vm{y}))$ are the corresponding singular values.

We will quantify the strength of the noise vector $\vm{\xi}$ through its Orlicz-2 norm.  For a random vector $\vm{z}$, we define
\[
	\|\vm{z}\|_{\psi_2} = \inf\left\{u>0: \E \left[\er^{\|\vm{z}\|_2^2/u^2}\right] \leq 2 \right\},
\]
and for scalar random variables we simply take $\vm{z} \in \reals^1$ in the expression above.  The Orlicz-2 norm is finite if the entries of $\vm{z}$ are subgaussian, and is proportional to the variance if the entries are Gaussian.  Our results treat the noise $\vm{\xi} \in \reals^\O$ as a random vector with iid entries that obey
\begin{align}
	\label{eq:CM-noise-stats}
	\|\xi[n]\|_{\psi_2} \leq \frac{\eta}{\O},\quad\text{and}\quad \|\vm{\xi}\|_{\psi_2} \leq c\eta. 
\end{align}
The following theorem states the stable recovery results for the KLT estimate. 
\begin{thm}
	\label{thm:CM-stablerec-CM1}
	Let $\vm{C}_0 \in \comps^{M \times W}$ be the rank-$R$ matrix of Fourier coefficients of an unknown signal ensemble, and let $\vm{y} = \cA(\vm{C}_0) + \vm{\xi}$ be noisy measurements taken by the M-Mux, where $\vm{\xi}$ obeys \eqref{eq:CM-noise-stats}.  If $\O \geq c\beta R(W+\mu_0^2 M)\log^2(MW)$ for some constant $\beta >1$, then the solution $\tilde{\vm{C}}$ to \eqref{eq:CM-KLT-est} will obey  
\begin{equation}
	\label{eq:CM-optimal-bound}
	\|\tilde{\vm{C}}-\vm{C}_0\|_{\F} \leq c\eta.
\end{equation} 
with probability at least $1-O(MW)^{-\beta}$.
\end{thm}

In contrast, we note that the result in \cite{candes10ma} could easily be adapted to show that under essentially the same conditions as Theorem~\ref{thm:CM-exactrec-CM1}, the solution $\tilde{\vm{C}}$ of \eqref{eq:CM-nuclearnorm_minnoisy} obeys
\[
	\|\tilde{\vm{C}}-\vm{C}_0\|_{\F}\leq c\sqrt{\min(W,M)} \eta.
\]
The above result is derived by only assuming that the noise $\vm{\xi}$ is bounded (i.e., $\|\vm{\xi}\|_2 \leq \eta$) with no statistical assumptions; see Lemma 1 in \cite{ahmed2012blind} for the proof.  Note that the result in \eqref{eq:CM-optimal-bound} is smaller by a factor of $1/\sqrt{\min(W,M)}$.

\subsection{Sampling Theorem for the FM-Mux}
\label{sec:sampling-thm-fm-mux}

As shown in Section~\ref{sec:CM-FM-Mux}, we can express the measurements taken by the FM-Mux in Figure~\ref{fig:CM-FM-Mux} as a linear operator $\cB: \comps^{M \times W} \rightarrow \reals^{\O}$ that maps the matrix of coefficients $\vm{C}_0$ to the samples $\vm{y}$. In this section, we present theory which demonstrated that a low rank $\vm{C}_0$ can be stably recovered using \eqref{eq:CM-nuclearnorm_minnoisy}.  We will establish this by showing that the linear operator $\cB$ satisfies the {\em restricted-isometry property} (RIP) for low-rank matrices.  The definition below is from \cite{recht10gu}:
\begin{defn}
	\label{defn:CM-mtx-RIP}
	A linear map $\cB: \comps^{M \times W} \rightarrow \reals^\Omega$ is said to satisfy the $R$-restricted isometry property if for every integer $1 \leq R \leq M$, we have a smallest constant  $\delta_R(\cB)$ such that 
	\[
	(1-\delta_R(\cB))\left\|\vm{C}\right\|_{\F} \leq \|\cB(\vm{C})\|_2 \leq (1+\delta_R(\cB))\left\|\vm{C}\right\|_{\F}
	\]
	for all matrices of  rank$(\vm{C}) \leq R$. 
\end{defn}
If $\delta_{2R}(\cB)<1$, then every rank-$R$ matrix $\vm{C}$ has a unique image through $\cB$.  If $\delta_{2R}(\cB)\leq 0.3$, then results from \cite{mohan10ne} show that given noisy measurements of an arbitrary matrix $\vm{C}_0$
\begin{equation}
 	\label{eq:CM-cB-meas-noisy}
 	\vm{y} = \cB(\vm{C}_0)+\vm{\xi},
\end{equation}
where $\left\|\vm{\xi}\right\|_2 \leq \eta$, the solution $\tilde{\vm{C}}$ to \eqref{eq:CM-nuclearnorm_minnoisy} satisfies
\begin{equation}
	\label{eq:matrix-stable}
	\|\tilde{\vm{C}}-\vm{C}_0\|_{\F} \leq c_*\frac{\|\vm{C}_0-\vm{C}_{0,R}\|_*}{\sqrt{R}} + c_{**} \eta.
\end{equation}
The matrix $\vm{C}_{0,R}$ above is the best rank-$R$ approximation to $\vm{C}_0$.  On contrary to our results for M-Mux in Theorem \ref{thm:CM-exactrec-CM1} and Theorem \ref{thm:CM-stablerec-CM1} that applied to the exact and stable recovery of strictly rank-$R$ matrix $\vm{C}_0$, the result in \eqref{eq:matrix-stable} applies to a general full-rank matrix $\vm{C}_0$ that could ideally be well approximated by a rank-$R$ matrix $\vm{C}_{0,R}$. In other words, the results apply to the recovery of a more general approximately correlated signal ensemble in \eqref{eq:CM-lowrankensemble}. An exact recovery result also follows from \eqref{eq:matrix-stable} by taking $\eta=0$ and $\vm{C}_0$ to be strictly rank $R$.

The following theorem, which we prove in Section~\ref{sec:CM-Theory3}, establishes the matrix RIP for the FM-Mux (and hence the accuracy in \eqref{eq:matrix-stable}) when the sampling rate $\O$ is within a logarithmic factor of $R(W+M)$.
\begin{thm}
	\label{thm:CM-exactstablerec-CM2}
	Let $\cB$ be the sampling operator for the FM-Mux, defined as in \eqref{eq:CM-PhiD} with sampling rate
	\[
		\Omega \geq c\beta R(M+W) \log^5(\Omega M)
	\]
	for a fixed constant $ c > 0$. Then $\delta_{2R}(\cB)\leq 0.3$ with probability at least $1-O(MW)^{-\beta}$, where $\beta > 0$ is a parameter that depends on $\delta_{2R}(\cB)$.
\end{thm}
As a consequence of this theorem, we can recover an ensemble of correlated signals $\vm{X}_c(t)$ by filtering, modulating, and sampling at a rate that scales linearly with $R$ and is within a constant and logarithmic factors of the optimal sampling rate.

\subsection{Non-periodic signals}
\label{sec:CM-nonperiodic}

The analysis in this paper depends on representing each signal in the ensemble using a Fourier series over the time interval $[0,1]$.  However, the recovery techniques (and most likely the analysis as well) can be extended to signals which are not periodic by windowing the input, and representing each interval of time using something akin to a short time Fourier transform.  For example, we might use a lapped orthogonal transform \cite{malvar89lo} to represent a non-periodic signal $x_m(t)$ for $t\in\R$:
\begin{equation}
	\label{eq:lot}
	x_m(t) = \sum_{n=-\infty}^\infty \sum_{\omega=0}^{W-1} \alpha_{m,k}[\omega]\psi_{n,\omega}(t), 
	\quad \text{where} \quad \psi_{n,\omega}(t) = g(t-n)\cos(\omega_k t).
\end{equation}
For a careful choice of (equally spaced) frequencies $\omega_k$ and smooth window $g(\cdot)$, the $\psi_{n,\omega}$ are orthonormal, and the notion of bandlimitedness corresponds roughly to choosing an $\Omega$.  The windows $g(\cdot-n)$ will overlap each other for consecutive $n$, meaning that some of the samples will be measuring multiple time-windows.  As such, the signals should be reconstructed over multiple time frames simultaneously, meaning the sum in \eqref{eq:lot} runs over a finite set of $n$ which includes every interval involved in a batch of samples.  We can then using a sliding window for the reconstruction, adding in the basis function representing the signal ensemble over are new interval of time, and removing intervals falling outside the window.  The solution inside the sliding window is updated constantly, with the previous solution serving as a ``warm start'' for the new optimization problem.

A framework similar to this for sparse recovery is described in detail in \cite{asif13sp}.

\section{Numerical Experiments}
\label{sec:CM-Exps}

This section presents a number of numerical experiments that illustrate the sampling performance of both compressive multiplexing architectures.  The experiments below measure the compression factor which can be achieved as a function of rank and accuracy.  We also run a stylized experiment using a data set obtained from an actual neural experiment. 

\subsection{Sampling performance}
\label{sec:CM-samp-perf}

In the experiments in this subsection and the next, the unknown-rank-$R$ matrix $\vm{C}_0$ is generated at random by the multiplication of a tall $M \times R$ and a fat $R \times W$ matrix, each with i.i.d. Gaussian entries.  This type of random matrix $\vm{C}_0$ of Fourier coefficients will correspond to a signal ensemble which is dispersed in time.  For these types of signals, we expect the M-Mux and the FM-Mux to have identical performance; as such, we will limit our simulation to the M-Mux architecture. We will call a reconstruction successful when its relative error is sufficiently small, specifically 
\[
	\text{Relative error} = \frac{\|\tilde{\vm{C}}-\vm{C}_0\|_{\F} }{\|\vm{C}_0\|_{\F}} \leq 10^{-3}.
\]
We will evaluate the sampling performance by trading off the sampling efficiency $\eta := R(W+M-R)/\O$ (or the oversampling factor $1/\eta$) against the compression factor $\gamma : = \O/(MW)$.  The success rate is computed over 100 iterations with different random instances of $\vm{C}_0$ in each iteration. 

In the first set of experiments, we take $M = 100$ signals, each bandlimited to $W/2 = 512$Hz.  The phase transition in Figure \ref{fig:CM-CompvsEff} relates the sampling efficiency with the compression factor. The shade represents the empirical probability of success. It is clear that the efficiency is high and improves further with increasing sampling rate. The phase transition in Figure \ref{fig:CM-RankvsOversamp} depicts the trend of the sampling rate for the successful recovery against the increasing rank. Interestingly, the sampling efficiency increases with the increasing values of $R$.  Under the same conditions, the plot in Figure \ref{fig:CM-RankvsComp} depicts the relationship between the lowest sampling rate $\O$, required for the $99\%$ success rate, and the number $R$ of independent signals. For clarity, the vertical axis shows the values of the compression factor instead of showing the plane sampling rate. It is evident that the sampling rate scales linearly with $R$ and is actually with in a small constant of the optimal sampling rate.
\begin{figure}[ht]
	\centering
	\subfigure[]{
        \includegraphics[trim=2.5cm 7.5cm 2.5cm 7.5cm,scale = 0.45]{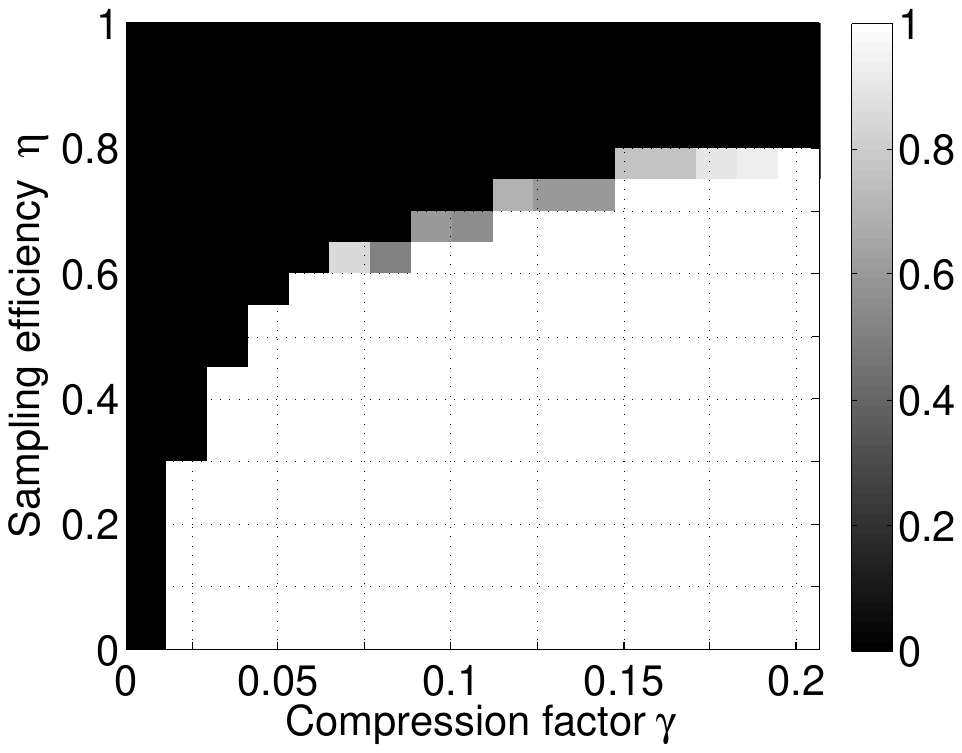}
        \label{fig:CM-CompvsEff}}
  	\subfigure[]{
	      \includegraphics[trim=2.5cm 7.5cm 2.5cm 7.5cm,scale = 0.45]{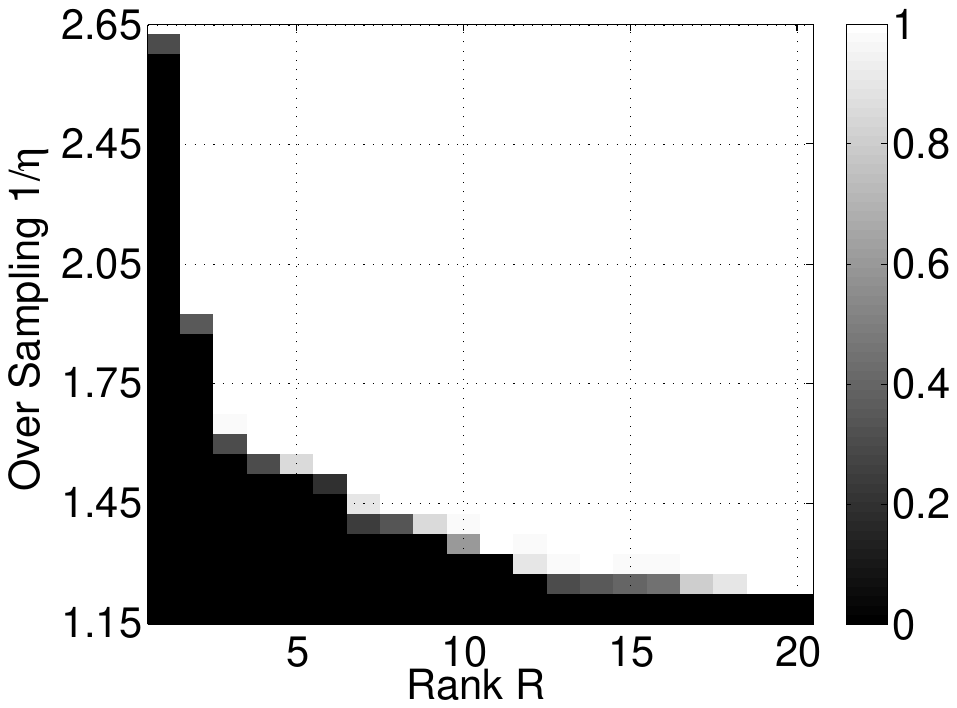}
	      \label{fig:CM-RankvsOversamp}}
        \caption{\small\sl Empirical probability of success for the compressive signal acquisition using the simulated M-Mux. In these experiments, we take an ensemble with $100$ signals, each bandlimited to $512$Hz. The shade shows the probability of success. (a) Success rate as a function of the compression factor and the sampling efficiency. (b) Success rate as a function of number of independent signals and the oversampling.} 
\label{fig:CM-phasetransitions}
\end{figure}

In the final experiment, we take $M = 20\alpha, W = 200\alpha$, and $R = 15$. The blue line in Figure \ref{fig:CM-AlphsvsSamples} illustrates the effect of varying the number of signals, and their bandwidth (by varying $\alpha$) on the minimum sampling rate required using the M-Mux for the successful reconstruction, while keeping fixed number $R$ of independent signals. For reference, the red line plots the corresponding cumulative Nyquist rate for each value of $\alpha$. The graph $\Omega$ depends linearly on $\alpha$, while cumulative Nyquist rate, of course, scales with $\alpha^2$. That is, the gap between $\O$ and the cumulative Nyquist rate widens very rapidly with increasing $M$ and $W$. The graph also shows that the sampling efficiency does not decrease much with the increasing $M$ and $W$. Hence, the sampling efficiency only depends on $R$. 
\begin{figure}[ht]
	\centering
	\subfigure[]{
    \includegraphics[trim=2.5cm 6.8cm 2.5cm 7.5cm,scale = 0.42]{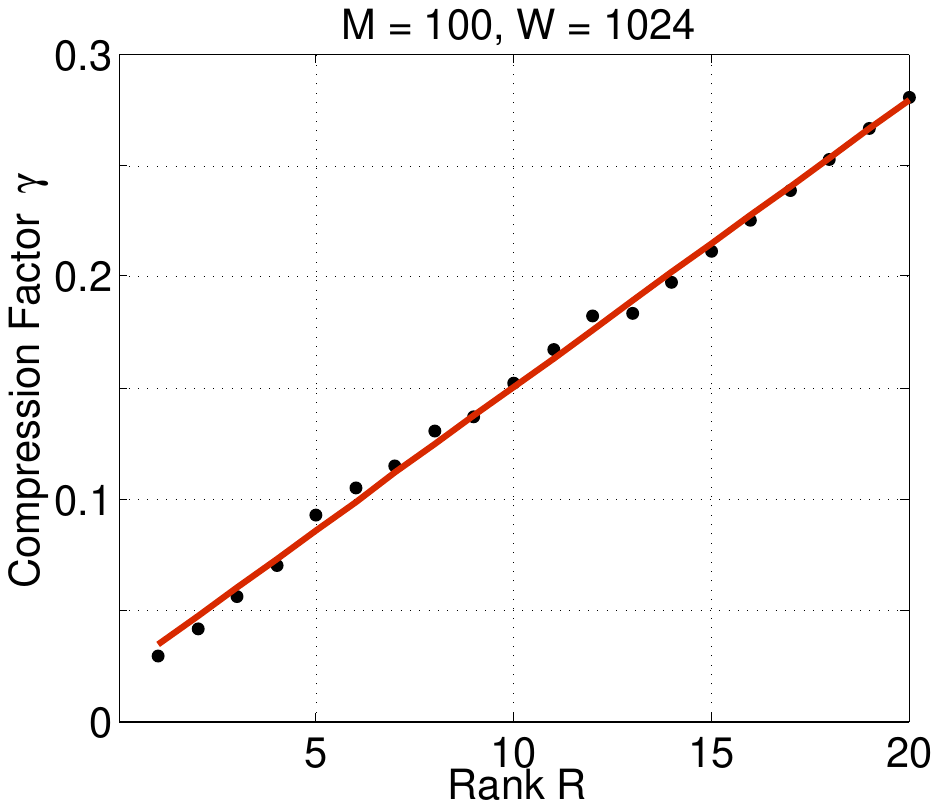}
    \label{fig:CM-RankvsComp}}
  \subfigure[]{
  \includegraphics[trim=2.5cm 7.5cm 2.5cm 7.5cm,scale = 0.45]{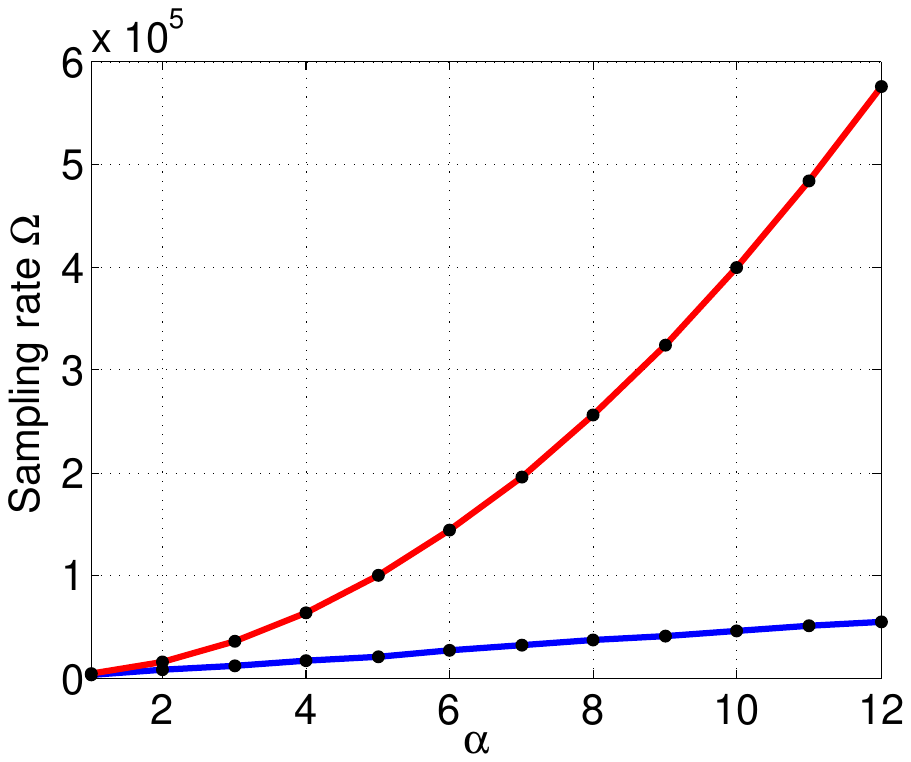}
    \label{fig:CM-AlphsvsSamples}}
\caption{\small\sl (a) Sampling as a function of number of independent signals. The simulated M-Mux takes an ensemble of $100$ signals, each bandlimited to $512$Hz. The discs mark the lowest sampling rate for the signal reconstruction with empirical success rate of 99$\%$. The vertical axis corresponds to $\gamma = \O/(MW)$, the ratio of the sampling rate to the cumulative Nyquist rate. The red line is the linear least squares fit of the data points. (b) Sampling rate as a function of $M$, and $W$. The simulated M-Mux takes an ensemble of $M = 20\alpha$ signals, each bandlimited to $W/2 = 100\alpha$Hz with number of underlying independent signals fixed at $R = 15$. The discs in the blue line mark the lowest sampling rate for the signal reconstruction with empirical success rate of 99$\%$. The red line shows the corresponding cumulative Nyquist rate. }
\end{figure}

\subsection{Recovery in the presence of noise}
\label{sec:CM-stable-rec}

This section simulates the performance of the multiplexer when are contaminated with additive noise $\vm{\xi} \sim \mathcal{N}(0,\sigma^2\vm{I})$ as in \eqref{eq:CM-noisy-meas}. For the signal reconstruction, we solve the optimization program \eqref{eq:CM-nuclearnorm_minnoisy} with $\delta = (\O + \sqrt{\O})^{1/2}\sigma$, a natural choice as $\|\vm{\xi}\|_2 \leq \delta$ holds with high probability. In all of the experiments in this section, we select $M = 100$, $W = 1024$, and $R = 15$.

Figure \ref{fig:CM-StableRec1} shows the signal-to-noise ratio (SNR) in dBs $(10\log_{10}( \|\vm{C}_0\|_{\F}^2/ \|\vm{\xi}\|_2^2))$ versus the relative error in dBs $(10\log_{10}( \mbox{(relative error)}^2 ))$. Each data point is generated by averaging over ten iterations, each time with independently generated matrices $\vm{C}_0$, and noise vector $\vm{\xi}$. The graph shows that the error increases gracefully as the SNR decreases. Figure \ref{fig:CM-StableRec2} depicts the decay of relative error with increasing sampling rate. 

\begin{figure}[ht]
	\centering
	\subfigure[]{
	      \includegraphics[trim=2.5cm 7.5cm 2.5cm 7.5cm,scale = 0.4]{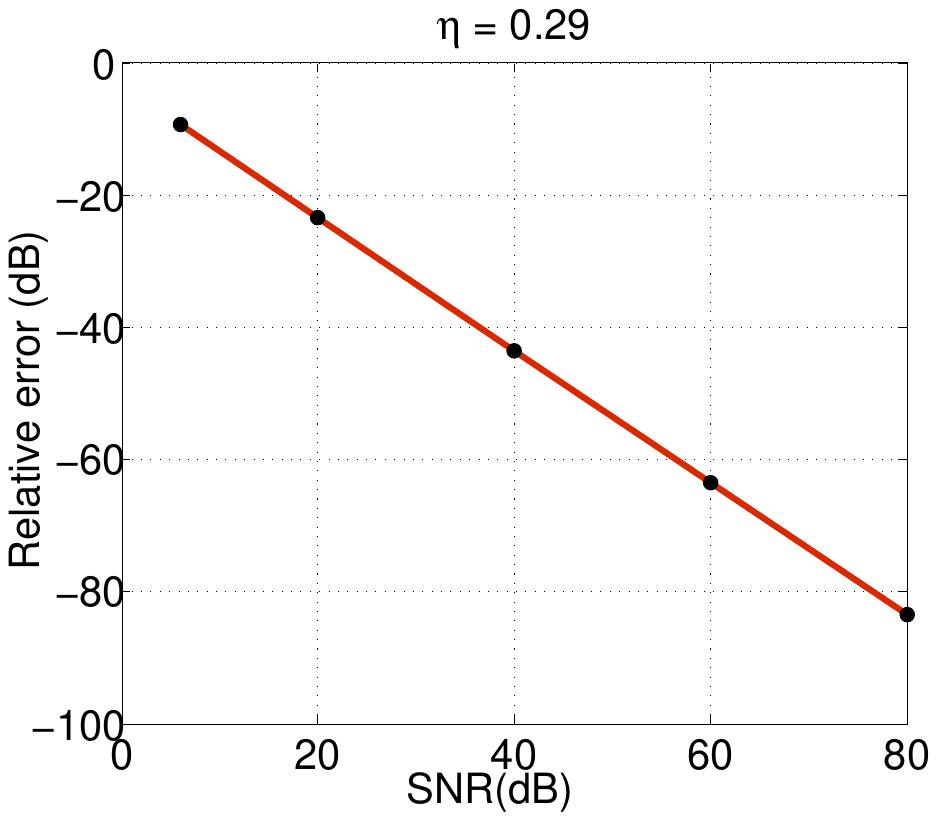}
	      \label{fig:CM-StableRec1}}
	\subfigure[]{
        \includegraphics[trim=2.5cm 7.5cm 2.5cm 7.5cm,scale = 0.4]{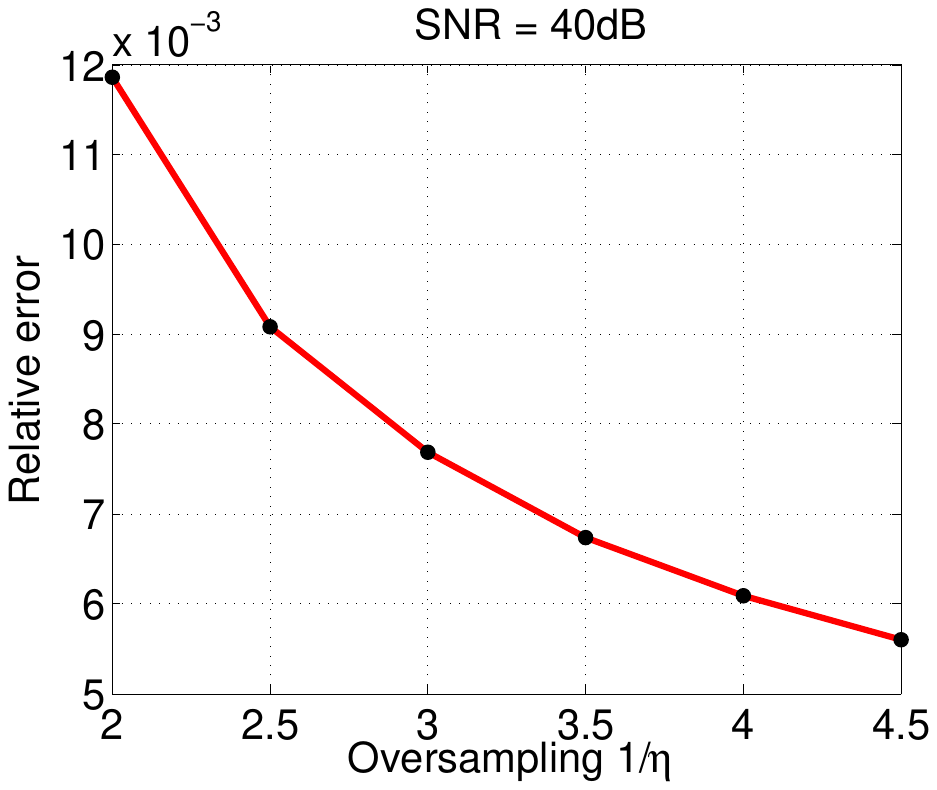}
        \label{fig:CM-StableRec2}}
\caption{\small\sl Recovery using the matrix Lasso in the presence of noise.  The input ensemble to the simulated M-Mux consists of $100$ signals, each bandlimited to $512$Hz with number $R = 15$ of latent independent signals.(a) The SNR in dB versus the relative error in dB.  The sampling rate is fixed and is given by the parameter $\eta = 0.29$. (b) Relative error as a function of the sampling rate. The SNR is fixed at 40dB.}
	\label{fig:CM-StableRec}
\end{figure}

The second set of experiments in this section, shown in Figure \ref{fig:CM-Comparison-KLT-LASSO}, depict the comparison between the performance of the matrix Lasso in \eqref{eq:CM-nuclearnorm_minnoisy}, and the one step thresholding KLT estimator in \eqref{eq:CM-KLT-est}. The first plot compares the two techniques for at an SNR = 40dB, meaning that there is very little noise contaminating the measurements.  It is clear that in this case the matrix Lasso outperforms the KLT estimator by considerable margin. The second plot shows that the reconstruction results are at least comparable in the presence of large (SNR = 6dB, 10dB) noise.  We see that while we can establish that the KLT estimator gives near-optimal results in theory, it is outperformed by the matrix Lasso in practice.

\begin{figure}[ht]
	\centering
	\subfigure[]{
        \includegraphics[trim=2.5cm 7.5cm 2.5cm 7.5cm,scale = 0.4]{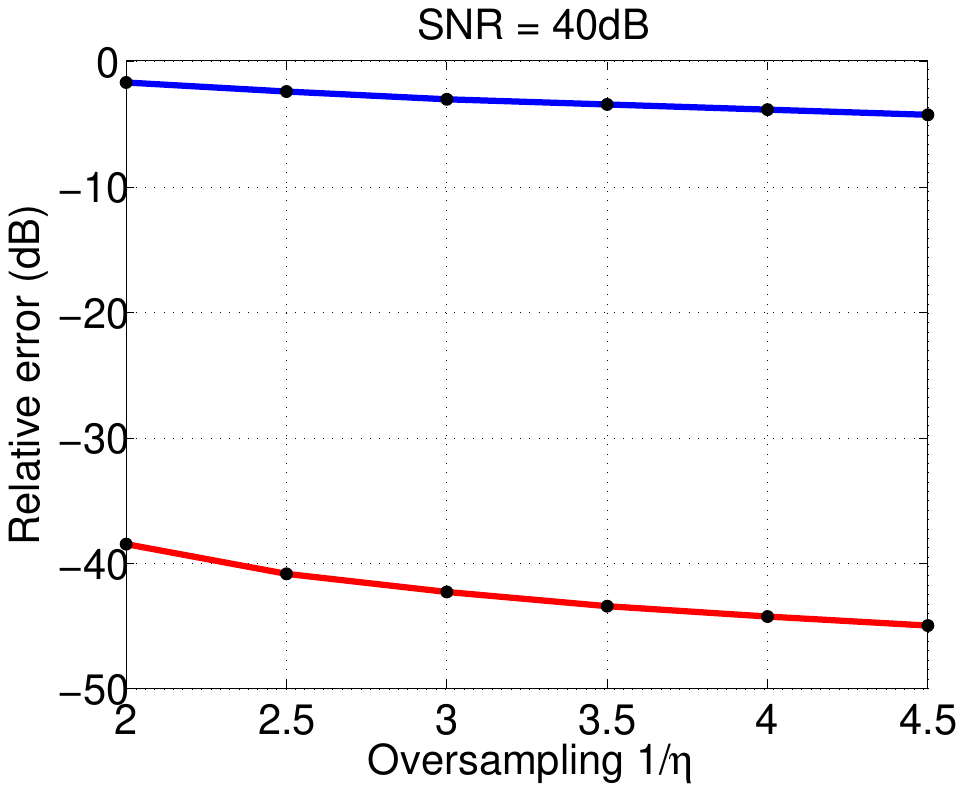}
        \label{fig:CM-StableRec3}}
        \subfigure[]{
        \includegraphics[trim=2.5cm 7.5cm 2.5cm 7.5cm,scale = 0.4]{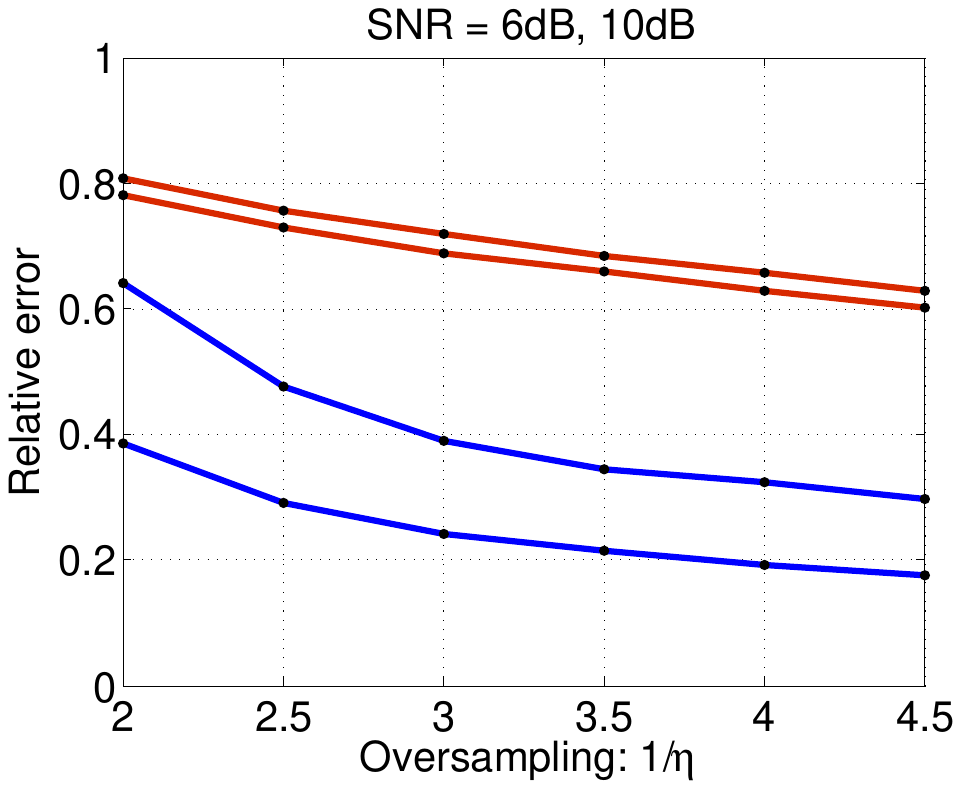}
        \label{fig:CM-StableRec4}}
\caption{\small\sl Comparison of the effectiveness of the matrix Lasso in \eqref{eq:CM-nuclearnorm_minnoisy} with KLT estimator in \eqref{eq:CM-KLT-est} for the signal reconstruction in the presence of noise.  The input ensemble to the simulated M-Mux consists of $100$ signals, each bandlimited to $512$Hz with number $R = 15$ of latent independent signals.(a) Relative error in dB versus oversampling factor; the red, and blue lines depict the performance of matrix Lasso, and the KLT estimator, respectively. The SNR is fixed at 40dB. (b) Relative error versus oversampling factor; the red and blue lines depict the performance of matrix Lasso and the KLT estimator, respectively. The plots are for the SNRs of 6dB and 10dB.}
	\label{fig:CM-Comparison-KLT-LASSO}
\end{figure}

\subsection{Neuronal experiment}

\begin{figure}[ht]
	\centering
  \includegraphics[trim=2.5cm 7.5cm 2.5cm 7.5cm,scale = 0.45]{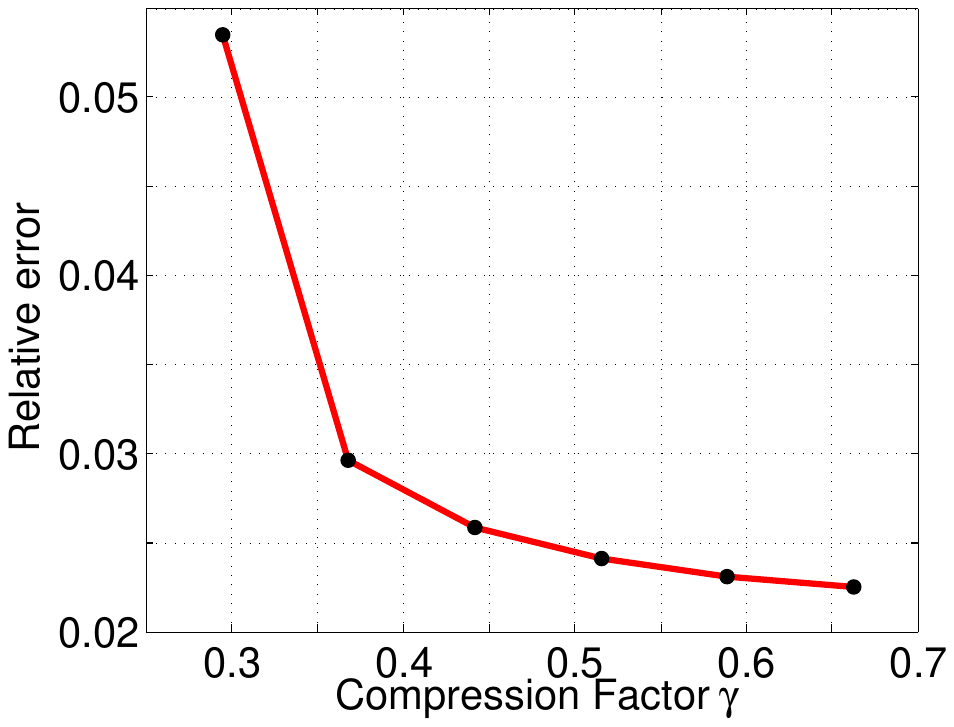}
  \label{fig:CM-NeuralMMux}
	\caption{\small\sl The performance of the M-Mux in an actual neural experiment. Compression factor as a function of the relative error. An ensemble of 108 signals recorded using polytrodes each required to be sampled at 100KHz is acquired using the M-Mux. Even by cutting the sampling rate in half the ensemble can be acquired with $97\%$ accuracy}
\end{figure}

In this subsection, we evaluate the performance of the M-Mux on the data set obtained from an actual neural experiment \cite{dataset} described in Section \ref{sec:CM-app}. We take neural signals recorded by two polytrodes containing a total of 108 recording sites. The signals recorded at each site are required to be sampled at 100,000 samples per second. That is, the Nyquist sampling rate for the acquisition of entire ensemble is 10.8 million samples per second. As mentioned earlier, the signals recorded from such micro sensor arrays are correlated, in particular, the $108 \times 1000$ matrix of samples over a window of 10ms can be approximated by a rank $R = 22$ matrix (to within a relative error of $0.018$). The result in Figure \ref{fig:CM-NeuralMMux} shows that we can reliably acquire the recorded ensemble for this application using the M-Mux at a smaller rate compared to the cumulative Nyquist rate. The compression factor is expected to drop further as the number of recording sites continue to increase.

\section{Proof of Theorem \ref{thm:CM-exactrec-CM1}: Exact recovery for the M-Mux}
\label{sec:CM-Theory1}

Let 
\begin{align}
	\label{eq:CM-C0-svd}
	\vm{C}_0 = \u\boldsymbol{\Sigma}\v^*
\end{align}
be the SVD of $\vm{C}_0$ and let $T$ be the linear space spanned by rank-one matrices of the form $\uu_r\vm{y}^*$ and $\vm{x}\vv_r^*$, $1 \leq r \leq R$, where $\vm{x}$ and $\vm{y}$ are arbitrary. The orthogonal projection of $\PT$ onto $T$ is defined as
\begin{equation}\label{eq:CM-PT-def}
\PT(\z) = \u\u^*\z + \z\v\v^*  - \u\u^*\z\v\v^*,
\end{equation}
and orthogonal projection $\PTc$ onto the orthogonal complement $T^{\perp}$ of $T$ is then 
\begin{align*}
\mathcal{P}_{T^\perp}(\z) = (\mathcal{I}-\PT)(\vm{Z}) = (\vm{I}_{M} - \u\u^*)(\vm{Z})(\vm{I}_{W} - \v\v^*),
\end{align*}
where $\vm{I}_d$ denotes the $d\times d$ identity matrix. It follows form the definition of $\PT$ that
\[
\PT(\A_\o) = (\u\u^*\vm{d}_\o)\ff_\o^* + \vm{d}_\o(\v\v^*\ff_\o)^*-(\u\u^* \vm{d}_\o)(\v\v^*\ff_\o)^*.
\]
Using \eqref{eq:CM-Aomega-def}, we have 
\begin{align}\label{eq:CM-PTAo-fro}
\left\|\PT(\A_\o)\right\|_{\F}^2  &= \left\langle \PT (\A_\o) , \A_\o\right\rangle\notag \\
& = \<\u\u^*\d_\o\ff_\o^*,\d_\o\ff_\o^*\> + \<\d_\o\ff_\o^*\v\v^*,\d_\o\ff_\o^*\>-\<\u\u^* \d_\o\ff_\o^*\v\v^*,\d_\o\ff_\o^*\>\notag\\
%& = \<\u^*\d_\o\ff_\o^*,\u^*\d_\o\ff_\o^*\> + \<\d_\o\ff_\o^*\v,\d_\o\ff_\o^*\v\>-\<\u^* \d_\o\ff_\o^*\v,\u^*\d_\o\ff_\o^*\v\>\\
& = \|\ff_\o\|_2^2 \|\u^*\d_\o\|_2^2 + \|\v^*\ff_\o\|_2^2\|\d_\o\|_2^2 - \|\u^*\d_\o\|_2^2 \|\v^*\ff_\o\|_2^2\notag\\
& \leq \frac{W}{\O}\|\u^*\d_\o\|_2^2 + M \|\v^*\ff_\o\|_2^2,
\end{align}
where the last inequality follows from the fact that $\|\u^*\d_\o\|_2^2 \|\v^*\ff_\o\|_2^2 \geq 0$, and  that $\|\vm{d}_{\o}\|_2^2 = M$, $\|\ff_\o\|_2^2 = \frac{W}{\O}$.

Standard results in duality theory for semidefinite programming assert that the sufficient conditions for the uniqueness of the minimizer of \eqref{eq:CM-nuclearnorm_min} are as follows:
\begin{itemize}
\item The linear operator $\cA$ is injective on the subspace $T$
\item $\exists \y \in \mbox{Range}(\cA^*)$, such that 
\begin{equation}\label{eq:CM-uniqueness-cond}
\|\PT(\y)-\u\v^*\|_{\F} \leq \frac{1}{2\sqrt{2}\gamma}, \quad \|\PTc(\y)\| \leq \frac{1}{2},
\end{equation}
\end{itemize}
where $\gamma : = \|\cA\|$. The above conditions are also referred to as inexact duality \cite{candes09ex,candes10ma}. The operator norm $\|\cA\|$ can be bounded with high probability using the matrix Chernoff bound \cite{tropp12us}. In particular, it can be shown-- using an argument similar to Lemma 1 of \cite{ahmed2012blind}-- that for some $\beta > 1$
\begin{equation}\label{eq:CM-gamma-bound}
\gamma \leq \sqrt{M\log(M^2\O W)}
\end{equation} 
with probability at least $1-O((MW)^{-\beta})$.

%%%%%%%%%%%%%%%%%%%%%%%%%%%%%%%%%%%%%%%%%%%%%%%%%%%%%%%%%%%%%%%%%%%%%%%%%%%%%%%%%%%%%%%%%%%%%%%%%%% 
\subsection{The golfing scheme for the M-Mux}
\label{sec:CM-golfing}

To prove the bounds in \eqref{eq:CM-uniqueness-cond}, we will use the standard golfing scheme \cite{gross11re}. We start with portioning $\O$ into $\k$ disjoint partitions $\{\Gamma_k\}_{ 1\leq k \leq \k}$, each of size $|\Gamma_k| = \Delta$, such that $\O = \Delta\k$. We take $\Gamma_k = \left\{ k+(j-1)\k: j \in \{1,\ldots,\Delta\}\right\}$. As will be shown later, we will be interested in knowing how closely the quantity $\E \cA_k^*\cA_k(\w)$ approximates $\w$. Suppose the measurements indexed by the set $\Gamma_k$ are provided by linear operator $\cA_k$, that is, 
\begin{align}
\cA_k(\w) = \{\tr{\ff_\o\d_\o^*\w}\}_{\o \in \Gamma_k}.\label{eq:CM-cA-iterate}
\end{align}
This means
$$\cA_k^*\cA_k(\w) = \sum_{\o \in \Gamma_k} \d_\o\d_\o^*\w\ff_\o\ff_\o^*,$$
which implies that 
$$\E \cA_k^*\cA_k(\w) = \sum_{\o \in \Gamma_k} \w\ff_\o\ff_\o^* = \frac{1}{\k}\w.$$
The last equality follows using the identities
$$\E \vm{d}_{\o}\vm{d}_{\o}^* = \vm{I}_M, \quad \sum_{\o \in \Gamma_k}\ff_\o\ff_\o^* = \frac{1}{\k}\vm{I},$$ 
where the second identity follows from the fact that the sub-matrix formed by selecting the columns of partial Fourier matrix $\tilde{\vm{F}}$ indexed by the set $\Gamma_k$ has orthogonal rows when $\Delta \geq W$; our later analysis will conform to this choice of partition size $\Delta$. Note that golfing scheme with index sets $\Gamma_k$ only works for the signals under consideration that are composed of first $W$ frequency components; see \eqref{eq:CM-C0-def}. In contrast to the signals with first $W$ active frequency components, we can extend the golfing argument to signals with $W$ active frequency components located anywhere in the set $\{1, \ldots,\O\}$; for details, see the golfing scheme in \cite{ahmed2012blind}. In other words, the M-Mux works equally well for the bandlimited signals regardless of the location of the active band in the total bandwidth $\O$.

We begin by iteratively constructing the dual certificate $\y \in \mbox{Range}(\cA^*)$ as follows. Let $\y_0 = 0$, and setup the iteration 
\begin{equation}
	\label{eq:CM-iteration}
	\y_k = \y_{k-1}+\k\cA_k^*\cA_k\left(\u\v^*-\PT(\y_{k-1})\right) \quad 
	\mbox{Note that} \quad \y_k \in \mbox{Range}(\cA^*),
\end{equation}
from which it follows that
\[
\PT(\y_k) = \PT(\y_{k-1})+\k\PT(\cA_k^*\cA_k)\PT\left(\u\v^*-\PT(\y_{k-1})\right);
\]
furthermore, define
\begin{align}\label{eq:CM-Wk-def}
\w_k :& = \PT (\y_k) -\u\v^*,
\end{align}
which gives an equivalent iteration
\begin{align}
\w_k &= \w_{k-1}-\k\PT\cA_k^*\cA_k\PT(\w_{k-1})\notag\\
&= \left(\PT-\k\PT\cA_k^*\cA_k\PT\right)(\w_{k-1})\label{eq:CM-W-iterate}.
\end{align}
Now the Frobenius norm of the iterates $\w_k$ is 
\[
\|\w_k\|_{\F} \leq \max_{ 1\leq k \leq \k}\|\PT-\k\PT\cA_k^*\cA_k\PT\|\|\w_{k-1}\|_{\F},
\]
which by repeated application of Lemma \ref{lem:CM-injectivity} gives a bound on the Frobenius norm of the iterates $\w_k$
\begin{align}\label{eq:CM-Wk-Fro-norm}
\left\|\w_{k}\right\|_{\F} &\leq \left(\frac{1}{2}\right)^{k}\left\|\u\v^*\right\|_{\F} = 2^{-k}\sqrt{R},~\text{for every} ~ k = 1,2,3,\ldots, \kappa
\end{align}
when $\O \geq c\beta \k R(\mu_0^2 M+W) \log^2(MW)$ with probability at least $1-O(\k(MW)^{-\beta})$. Hence, the final iterate obeys
\begin{align}\label{eq:kappa-size}
\left\|\w_{\k}\right\|_{\F} &\leq \frac{1}{2\sqrt{2}\gamma}, \quad \mbox{when} \quad \k \geq 0.5\log_2 (8\gamma^2 R)
\end{align}
with probability at least $1-O(\k(MW)^{-\beta})$. This proves the first bound in \eqref{eq:CM-uniqueness-cond}. In light of \eqref{eq:CM-coherence}, the coherence $\mu_k^2$ of $k$th iterate $\w_{k}$ is defined as
\begin{align}\label{eq:CM-Wk-coherence}
\mu_{k}^2 := \frac{\O}{R}\max_{\o \in \Gamma_k} \|\w_k\ff_\o\|_2^2. 
\end{align}
Lemma \ref{lem:CM-coherence-iterates} will show that $\mu_k^2 \leq 0.5\mu_{k-1}^2$, for every $k = 1,2 ,3, \ldots, \kappa$, which implies that
\begin{equation}\label{eq:CM-bound-coher}
\mu_k^2 \leq 2^{-k}\mu_0^2, \quad \text{for every}~ k \in \{1, \ldots, \k\}
\end{equation}
holds with probability at least $1-O(\O(MW)^{-\beta})$. The final iterate $\y_{\k} = -\sum_{k = 1}^{\k} \k\cA_{k}^*\cA_{k}\w_{k-1}$ of the iteration \eqref{eq:CM-iteration} will be our choice of the dual certificate. We will now show that $\y_{\k}$ obeys the conditions \eqref{eq:CM-uniqueness-cond}.   
\begin{align*}
\|\PTc(\y_{\k})\| &\leq \sum_{k = 1}^{\k} \|\PTc(\k\cA_k^*\cA_k\w_{k-1})\| = \sum_{k = 1}^{\k} \|\PTc(\k\cA_k^*\cA_k\w_{k-1}-\w_{k-1})\|\\
& \leq\sum_{k = 1}^{\k} \|(\k\cA_k^*\cA_k-\mathcal{I})\w_{k-1}\|_{\F} \leq \sum_{k = 1}^{\k} \max_{1 \leq k \leq \k} \|(\k\cA_k^*\cA_k-\mathcal{I})\w_{k-1}\|_{\F}\\
& \leq \sum_{k =1}^{\k} 2^{-k-1} < \frac{1}{2},
\end{align*}
where the third inequality holds with probability at least $1-O(\k(MW)^{-\beta})$ when
\[
\O \geq c\beta \k R \max(\mu_{0}^2 M,W ) \log^2(MW),
\] 
which is implied by Lemma \ref{lem:CM-concentration}, and Equation \eqref{eq:CM-bound-coher}. We pick $\kappa \leq c \log(MW)$ with a constant $c$ chosen such that \eqref{eq:kappa-size} is satisfied. Combining all these results and the probabilities gives us the conclusion of Theorem \ref{thm:CM-exactrec-CM1} with probability at least $1-O(\O(MW)^{-\beta})$. Since the sampling architectures are only interesting when the sampling rate is sub-Nyquist, i.e., $\O \leq MW$, we will simplify the success probability to  $1-O((MW)^{1-\beta})$.

\subsection{Main lemmas for Theorem~\ref{thm:CM-exactrec-CM1}}
\begin{lem}\label{lem:CM-injectivity}
Let $\cA_k$ be as defined in \eqref{eq:CM-cA-iterate}, and $\k$ be the number of partitions used in the golfing scheme; see Section \ref{sec:CM-golfing}. Then for all $\beta > 1$, 
\[
\max_{1\leq k \leq \k}\|\k\PT\cA_k^*\cA_k\PT-\PT\|\leq \frac{1}{2}
\]
provided $\O \geq c \beta\k R(\mu^2_0 M+W)\log^2(MW)$ with probability at least $1-O(\k(MW)^{-\beta})$.
\end{lem}
\begin{lem}\label{lem:CM-concentration}
Let $\mu_{k-1}^2$, as in \eqref{eq:CM-Wk-coherence} be the coherence of the iterate $\w_{k-1}$, defined in \eqref{eq:CM-Wk-def}. Then for all $\beta > 1$ 
\[
\max_{1 \leq k\leq \k} \|\k\cA_k^*\cA_k(\w_{k-1}) - \w_{k-1}\| \leq 2^{-k-1}
\]
with probability at least $1-O(\k(MW)^{-\beta})$ provided $\O \geq c\beta\k\max(W,\mu_0^2 M)\log^2(MW)$, where $\k$ is the total partitions used in the golfing scheme.
\end{lem}
\begin{lem}\label{lem:CM-coherence-iterates}
Let $\w_k$, and $\mu_{k}^2$ be as in \eqref{eq:CM-Wk-def}, and \eqref{eq:CM-Wk-coherence}.  
If $\O \geq c\k \beta R(\mu_0^2M+W)\log^2(MW)$, then 
\[
\mu_k^2  \leq \frac{1}{2} \mu_{k-1}^2
\]
holds for every $k = 1,2,3,\ldots,\kappa$ with probability at least $1-O(\O(MW)^{-\beta})$. The number $\k$ is the total partitions used in golfing in Section \ref{sec:CM-golfing}. 
\end{lem}

Finally, we will use a specialized version of the matrix Bernstein-type inequality \cite{tropp12us,koltchinskii10nu} to bound the operator norm of the random matrices in this paper. The version of Bernstein listed below depends on the Orlicz norms $\|\z\|_{\psi_\alpha},~ \alpha \geq 1$ of a matrix $\z$. The Orlicz norms are defined as 
\begin{equation}\label{eq:matpsinorm}
\|\z\|_{\psi_\alpha} = \inf \{ u>0: \E \exp (\frac{\|\z\|^\alpha}{u^\alpha}) \leq 2\}, \quad \alpha \geq 1.
\end{equation}
Suppose that, for some constant $U_{\alpha} > 0, \|\z_q\|_{\psi_\alpha} \leq U_{(\alpha)}, q = 1, \ldots, Q$ then the following proposition holds. 
\begin{prop}\label{prop:matbernpsi}
Let $\vm{Z}_1,\vm{Z}_2,\ldots,\vm{Z}_Q$ be iid random matrices with dimensions $M \times N$ that satisfy $\E (\vm{Z}_q) = 0$. Suppose that $\|\z\|_{\psi_\alpha} < \infty$ for some $\alpha \geq 1 $. Define 
\begin{align}\label{eq:matbernsigma}
&\sigma_Z  = \max \left\{\left\|\sum_{q = 1}^Q (\E \vm{Z}_q\vm{Z}_q^*)\right\|^{1/2},\left\|\sum_{q = 1}^Q (\E \vm{Z}_q^*\vm{Z}_q)\right\|^{1/2} \right\}
\end{align}
Then $\exists$ a constant $C > 0$ such that , for all $t>0$, with probability at least $1-\er^{-t}$
\begin{align}\label{eq:matbernpsi}
&\left\|\vm{Z}_1+\cdots+\vm{Z}_Q\right\| \leq C \max\left\{\sigma_Z\sqrt{t+\log(M+N)},U_{\alpha}\log^{1/\alpha}\left(\frac{QU_\alpha^2}{\sigma_Z^2}\right)(t+\log(M+N))\right\}
\end{align}
\end{prop}

\subsection{Proof of Lemma \ref{lem:CM-injectivity}}
In this section, we are concerned with bounding the centered random process 
\begin{align*}
\k\PT\cA_k^*\cA_k \PT - \PT  &= \k\PT\cA_k^*\cA_k \PT- \E \k\PT (\cA_k^*\cA_k) \PT\\
&= \k\sum_{\o \in \Gamma_k} (\PT(\A_\o)\otimes\PT(\A_\o)-\E \PT(\A_\o)\otimes\PT(\A_\o)),
\end{align*}
where we have used the fact 
$$\k \E \PT\cA_k^*\cA_k \PT = \k\PT\E (\cA_k^*\cA_k) \PT = \PT$$
The last equality follows from the fact that 
$$\E (\cA_k^*\cA_k) = \frac{1}{\k}\mathcal{I}.$$
Now define $\cL_\o$, which maps $\vm{C}$ to $\< \PT(\A_\o),\vm{C}\> \PT(\A_\o)$. This operator is rank-1 with operator norm $\|\cL_\o\| = \|\PT(\A_\o)\|_{\F}^2$, and we are interested in bounding the operator norm
\[
\|\k\PT\cA_k^*\cA_k\PT-\PT\| = \k\sum_{\o \in \Gamma_k} (\cL_\o-\E\cL_\o)
\] 
For this purpose, we will use matrix Bernstein's bound in Proposition \ref{prop:matbernpsi}. Since $\cL_\o$ is symmetric, we only need to calculate the following for variance 
\begin{align}\label{eq:CM-variance-injectivity}
\k^2\|\sum_{\o \in \Gamma_k} \E \cL_\o^2 -(\E\cL_\o)^2\| \leq \k^2\|\sum_{\o \in \Gamma_k} \E \cL^2_\o\| = \k^2\|\E \sum_{\o \in \Gamma_k}\|\PT(\A_\o)\|_{\F}^2 \cL_\o\|,
\end{align}
where the inequality follows from the fact that $\E \cL_\o^2$, and $(\E\cL_\o)^2$ are symmetric positive-semidefinite (PSD) matrices, and for PSD matrices $\A$, and $\vm{B}$, we have $\|\A-\vm{B}\| \leq \max\{\|\A\|,\|\vm{B}\|\}$. Plugging in the definition of $\cL_\o$ and using \eqref{eq:CM-PTAo-fro}, we have 
\begin{align}
\left\|\E \sum_{\o \in \Gamma_k}\|\PT(\A_\o)\|_{\F}^2 \cL_\o\right\| &\leq \left\|\sum_{\o \in \Gamma_k} \E\left\{\left(\frac{W}{\O}\|\u^*\d_\o\|_2^2+M \|\v^*\ff_\o\|_2^2\right)\cL_\o\right\}\right\|\notag\\
&\leq \frac{W}{\O}\left\|\E \sum_{\o \in \Gamma_k} \|\u^*\d_\o\|_2^2\cL_\o\right\|+M\left(\max_{1 \leq \o \leq \O}\|\v^*\ff_\o\|_2^2\right)\left\|\sum_{\o \in \Gamma_k} \E \cL_\o \right\|\notag\\
&\leq \frac{W}{\O}\left\|\E \sum_{\o \in \Gamma_k} \|\u^*\d_\o\|_2^2\cL_\o\right\|+\mu_0^2R\frac{M}{\O}\left\|\sum_{\o \in \Gamma_k} \E \cL_\o \right\|.\label{eq:CM-Main}
\end{align}
The last inequality follows form the definition of the coherence \eqref{eq:CM-coherence}. Before proceeding further, we write out the tensor $\A_{\o} \otimes \A_{\o}$ in the matrix form:
\begin{align*}
\d_\o\ff_\o^*\otimes \d_\o\ff_\o^* &= \begin{bmatrix}
d_\o[1]d_\o[1] \ff_\o\ff_\o^* & d_\o[1]d_\o[2]\ff_\o\ff_\o^* & \cdots & d_\o[1]d_\o[M]\ff_\o\ff_\o^*\\
d_\o[2]d_\o[1] \ff_\o\ff_\o^* & d_\o[2]d_\o[2]\ff_\o\ff_\o^* & \cdots & d_\o[2]d_\o[M]\ff_\o\ff_\o^*\\
\vdots & \vdots & \ddots & \vdots\\
d_\o[M]d_\o[1] \ff_\o\ff_\o^* & d_\o[M]d_\o[2]\ff_\o\ff_\o^* & \cdots & d_\o[M]d_\o[M]\ff_\o\ff_\o^*\\
\end{bmatrix}\\
 & = \{d_\o[\alpha]d_\o[\beta]\ff_\o\ff_\o^*\}_{(\alpha,\beta)}.
\end{align*}
We will use $\bar{\uu}_{\alpha}$ to denote the $\alpha$th row of the matrix $\u$, and $\delta_{x}$ is the indicator function when the condition $x$ is true. Using these notations, we can simplify the following quantity of interest  
\begin{align*}
\|\E \|\u^*\d_\o\|_2^2 (\PT(\A_\o) \otimes \PT(\A_\o))\| &\leq \|\PT\| \|\E \|\u^*\d_\o\|_2^2 (\A_\o \otimes \A_\o)\| \|\PT\|\\
& \leq \|\E \|\u^*\d_\o\|_2^2 \{d_\o[\alpha]d_\o[\beta]\ff_\o\ff_\o^*\}_{\alpha,\beta}\|\\
&= \left\|\{ \|\u\|_{\F}^2\ff_\o\ff_\o^*\delta_{(\alpha=\beta)}+2\<\bar{\uu}_\alpha,\bar{\uu}_\beta\> \ff_\o\ff_\o^*\delta_{(\alpha\neq \beta)}\}_{(\alpha,\beta)}\right\|,
\end{align*}
where second inequality follows form the fact that $\|\PT\|\leq 1$ and the third equality follows by expanding and taking expectation on each entry of the matrix. Summing over $\o \in \Gamma_k$ gives 
\begin{align*}
\left\|\E \sum_{\o \in \Gamma_k} \|\u^*\d_\o\|_2^2 \PT(\A_\o) \otimes \PT(\A_\o)\right\| &\leq \left\|\sum_{\o \in \Gamma_k} \{\|\u\|_{\F}^2\ff_\o\ff_\o^*\delta_{(\alpha=\beta)}+2\<\bar{\uu}_\alpha,\bar{\uu}_\beta\>\ff_\o\ff_\o^*\delta_{(\alpha\neq\beta)}\}_{(\alpha,\beta)}\right\|\\
& =  \left\|\{ \|\u\|_{\F}^2\sum_{\o \in \Gamma_k} \ff_\o\ff_\o^*\delta_{(\alpha=\beta)}+2\<\bar{\uu}_\alpha,\bar{\uu}_\beta\> \sum_{\o \in \Gamma_k}\ff_\o\ff_\o^*\delta_{(\alpha\neq\beta)}\}_{(\alpha,\beta)}\right\| \\
& = \left\|\{ \|\u\|_{\F}^2\frac{1}{\k}\vm{I}_W\delta_{(\alpha=\beta)}\}_{(\alpha,\beta)}+\{2\<\bar{\uu}_\alpha,\bar{\uu}_\beta\> \frac{1}{\k}\vm{I}_W\delta_{(\alpha\neq\beta)}\}_{(\alpha,\beta)}\right\|.
\end{align*}
Now, it follows by simple linear algebra 
\begin{align*}
\left\|\E \sum_{\o \in \Gamma_k} \|\u^*\d_\o\|_{\F}^2 \PT(\A_\o) \otimes \PT(\A_\o)\right\| \leq \frac{1}{\k}(\|\u\|_{\F}^2 + 2\|\u\u^*\|) \leq \frac{R+2}{\k}.
\end{align*}
Plugging the above result, together with \eqref{eq:CM-Main} in \eqref{eq:CM-variance-injectivity}, we obtain  
\begin{equation}\label{eq:CM-var-injectivity}
\sigma_Z^2 = \k^2\left\|\sum_{\o \in \Gamma_k} \E [(\cL_\o -\E\cL_\o )^2]\right\| \leq c\k R\frac{\mu_0^2M+W}{\O}.
\end{equation}
Using the definition matrix Orlicz norm \eqref{eq:matpsinorm}, and the fact that $\cL_\o$, and $\E\cL_\o$ are positive semidefinite matrices, it follows
\begin{align}\label{eq:CM-orliczbound}
\k\|\cL_\o - \E [\cL_\o]\|_{\psi_1} & \leq \k\max\{\|\cL_\o\|_{\psi_1}, \|\E\cL_\o\|_{\psi_1}\} 
\end{align}
As shown earlier, we have $\|\cL_{\o}\| = \|\PT(\A_{\o})\|_{\F}^2$, and also it is easy to show that $\|\E\cL_{\o}\| = W/\O$. Using it together with \eqref{eq:matpsinorm}, and \eqref{eq:CM-PTAo-fro}, we obtain the Orlicz-1 norm 
\begin{align*}
\k\|\cL_{\o}\|_{\psi_1} &\leq \mu_0^2\k R\frac{M}{\O} + \k\frac{W}{\O}\left\|\sum_{r = 1}^R \left(\sum_{m = 1}^M d_\o[m]U[m,r]\right)^2\right\|_{\psi_1}\\
& \leq \mu_0^2\k R\frac{M}{\O} + \k\frac{W}{\O}\sum_{r = 1}^R \left\|\left(\sum_{m = 1}^M d_\o[m]U[m,r]\right)^2\right\|_{\psi_1}
\end{align*}
It can easily be shown that random variable: 
\[ 
Y = \sum_{m = 1}^M d_\o[m] U[m,r]
\]
is subgaussian, which implies that $Y^2$ is a sub-exponential random variable; see Lemma \ref{lem:CM-Orlicz-square}. In addition, by the independence of $\{d_{\o}[m]\}_{1\leq m \leq M}$ and using Lemma \ref{lem:CM-Orlicz-independence}, we have
\begin{align*}
\sum_{r = 1}^R\left\| \left(\sum_{m = 1}^M d_\o[m] U[m,r]\right)\right\|_{\psi_2}^2& \leq  c\sum_{r=1}^R \sum_{m = 1}^M \|d_\o[m]U[m,r]\|_{\psi_2}^2\leq cR.\\
\end{align*}
Hence, 
\[
\k\|\cL_{\o}\|_{\psi_1} \leq \mu_0^2\k R\frac{M}{\O} + cR\k\frac{W}{\O},
\]
which dominates the maximum in \eqref{eq:CM-orliczbound}, and thus $\k\|\cL_\o - \E [\cL_\o]\|$ is sub-exponential; hence, $\alpha = 1$ in \eqref{eq:matbernpsi}. Let $\Lambda = \mu_0^2M+W$, and as defined earlier that $|\Gamma_k|= \Delta$, and $\k = \O/\Delta$. Then
\begin{equation}\label{eq:CM-orlicznorm-injectivity}
U_{1}\log\left(\frac{|\Gamma_k| U_1^2}{\sigma^2_Z}\right) \leq c\k R\frac{\Lambda}{\O}\log(R\Lambda) 
\end{equation}
Plugging \eqref{eq:CM-var-injectivity}, and \eqref{eq:CM-orlicznorm-injectivity} in \eqref{eq:matbernpsi}, we have
\begin{align*}
\|\k\PT\cA_k^*\cA_k\PT-\PT\| \leq c\max\left\{\sqrt{\frac{\k R \Lambda\beta\log(MW)}{\O}},  \frac{\k R \Lambda}{\O}\log(R\Lambda)\beta \log(MW) \right\}.
\end{align*} 
The result of the Lemma \ref{lem:CM-injectivity} follows by taking  $\O \geq c\beta \k R \Lambda \log(MW)\log(R\Lambda)$, $t = \beta\log(MW)$, and using the union bound over $\k$ independent partitions.

\subsection{Proof of Lemma \ref{lem:CM-concentration}} 
We are interested in controlling the operator norm of 
\begin{equation}\label{eq:CM-objective}
\k\cA_k^*\cA_k(\w_{k-1}) -\w_{k-1} = \sum_{\o \in \Gamma_k} \k(\<\w_{k-1},\A_\o\> \A_\o - \E \<\w_{k-1},\A_\o\>\A_\o).
\end{equation}
To control the operator norm of the sum of random matrices 
\[
\z_\o = \k( \<\w_{k-1},\A_\o\>\A_\o - \E \<\w_{k-1},\A_\o\>  \A_\o)
\]
on the r.h.s. of \eqref{eq:CM-objective}, we will again refer to Proposition \ref{prop:matbernpsi}. We begin by evaluating the first variance term 
\begin{align*}
\left\|\sum_{\o \in \Gamma_k} \E \z_\o\z_\o^* \right\| &\leq \k^2\left\|\sum_{\o \in \Gamma_k} \E |\<\w_{k-1},\A_\o\>|^2 \A_\o\A_\o^*\right\| = \k^2\max_{1 \leq \o \leq \Omega}\|\ff_\o\|^2 \left\|\sum_{\o \in \Gamma_k} \E |\<\w_{k-1},\A_\o\>|^2 \vm{d}_\o \vm{d}_\o^*\right\|,
\end{align*}
where last equality follows form \eqref{eq:CM-Aomega-def}. Lemma \ref{lem:CM-Supporting1} shows that 
\[
\E |\<\w_{k-1},\A_\o\>|^2\vm{d}_\o \vm{d}_\o^* \preccurlyeq 3\|\w_{k-1}\ff_{\o}\|_2^2 \vm{I}_M.
\]
Summation over $\o \in \Gamma_k$ gives 
\[
\sum_{\o \in \Gamma_k} \E |\<\w_{k-1},\A_\o\>|^2\vm{d}_\o \vm{d}_\o^* \preccurlyeq \frac{3}{\k}\|\w_{k-1}\|_{\F}^2 \vm{I}_M,
\]
which implies that 
\begin{equation}\label{eq:CM-conc-var1}
\left\|\sum_{\o \in \Gamma_k}\E \z_\o\z_\o^*\right\| \leq 3\k\frac{W}{\O}\|\w_{k-1}\|_{\F}^2 \leq 3\kappa R\frac{W}{\Omega} 2^{-2(k-1)},
\end{equation} 
where the last inequality is the result of \eqref{eq:CM-Wk-Fro-norm}. The second variance term needs  
\begin{align*}
\left\|\sum_{\o \in \Gamma_k} \E \z_\o^*\z_\o\right\| &\leq \k^2\left\|\sum_{\o \in \Gamma_k} \E |\<\w_{k-1},\A_\o\>|^2\vm{A}_\o^*\vm{A}_\o\right\| \leq M \k^2\|\sum_{\o \in \Gamma_k} \ff_\o\ff_\o^*\| \max_{\o}\E |\<\w_{k-1},\A_\o\>|^2.
\end{align*}
Using the facts that $\E |\<\w_{k-1},\A_\o\>|^2 = \|\w_{k-1}\ff_{\o}\|_2^2,$ and $\sum_{\o \in \Gamma_k} \ff_\o\ff_\o^* = (1/\k)\vm{I}_{W}$ gives  
\begin{align}\label{eq:CM-conc-var2}
\left\|\sum_{\o \in \Gamma_k} \E \z_\o^*\z_\o\right\| &\leq \mu_{k-1}^2 \k R\frac{M}{\O} \leq \mu_{0}^2 \k R\frac{M}{\O}2^{-2(k-1)},
\end{align}
which follows by \eqref{eq:CM-Wk-coherence}. Plugging \eqref{eq:CM-conc-var1}, and \eqref{eq:CM-conc-var2} in \eqref{eq:matbernsigma}, we obtain 
\begin{equation}\label{eq:CM-variance-conc}
\sigma_{Z} \leq c 2^{-(k-1)}\cdot\max\left\{\sqrt{\mu_{0}^2\k R\frac{M}{\O}} , \sqrt{3\k R\frac{W}{\O}}\right\}.
\end{equation}
The fact that $\z_{\o}$ are subgaussian can be proven by showing that $\|\z_{\o}\|_{\psi_2} < \infty$. First, note that \begin{align*}
\|\z_\o\|_{\psi_2} &\leq 2 \|\k\<\w_{k-1},\A_\o\> \A_\o\|_{\psi_2}.
\end{align*}
Second, the operator norm of the matrix under consideration is 
\[
\|\<\w_{k-1},\A_\o\> \A_\o\| \leq \sqrt{\frac{MW}{\O}}|\<\w_{k-1},\A_\o\>|.
\]
Using the definition \eqref{eq:matpsinorm}, we obtain
\begin{align*}
\|\z_\o\|_{\psi_2} &\leq 2\k\sqrt{\frac{MW}{\O}} \|\<\w_{k-1},\A_\o\> \|_{\psi_2}.
\end{align*}
Let $\vm{w}_m^*$ denote the rows of the $M \times W$ matrix $\w_{k-1}$. We can write
\[
\<\w_{k-1},\A_\o\> = \sum_{m = 1}^M d_\o[m]\vm{w}_m^*\ff_{\o},
\]
and using the independence of $d_{\o}[m]$ with Lemma \ref{lem:CM-Orlicz-independence}, we see that
\begin{align*}
\|\<\w_{k-1},\A_\o\>\|_{\psi_2}^2 \leq c \sum_{m=1}^M \|\vm{w}_m^*\ff_\o\|_{\psi_2}^2 \leq c\mu_{k-1}^2\frac{R}{\O} \leq c 2^{-2(k-1)} \mu_0^2 \frac{R}{\Omega}. 
\end{align*}
Hence, $U_2$ in Proposition \ref{prop:matbernpsi} is 
\begin{align*}\label{eq:CM-Zw-orlicz-concentration}
U_2 = \|\z_\o\|_{\psi_2} \leq c  2^{-k+1}\left(\k^2\mu^2_0 R\frac{MW}{\O^2}\right)^{1/2} \leq c 2^{-k+1} \sqrt{\frac{\kappa}{\Omega}}\max\left\{ \sqrt{\mu_0^2 \kappa R\frac{M}{\Omega}}, \sqrt{\kappa R\frac{W}{\Omega}}\right\},
\end{align*}
and using the fact that $\k = \O/\Delta$, and $|\Gamma_k| = \Delta$, we obtain
\begin{align}
U_2\log^{1/2}\left(\frac{|\Gamma_k| U_2^2}{\sigma_Z^2} \right)  \leq c 2^{-k+1} \max\left\{ \sqrt{\mu_0^2 \kappa^2 R\frac{M}{\Omega}}, \sqrt{\kappa^2 R\frac{W}{\Omega}}\right\}\log^{1/2}\Omega.
\end{align}
Using \eqref{eq:CM-variance-conc}, and \eqref{eq:CM-Zw-orlicz-concentration} in \eqref{eq:matbernpsi} with $t = \beta\log(MW)$, we have 
\begin{align}
&\|\k\cA_k^*\cA_k(\w_{k-1})-\w_{k-1}\| \notag \\
&\quad \leq  c 2^{-k+1} \max\left\{ \sqrt{\mu_0^2 \kappa R\frac{M}{\Omega}}, \sqrt{\kappa R\frac{W}{\Omega}}\right\}\cdot \max\left\{\log^{1/2}(WM), \sqrt{\frac{\kappa}{\Omega} }\log^{1/2}\Omega \log (WM)\right\}.\label{eq:CM-concentration-bound}
\end{align}
Using \eqref{eq:CM-Wk-Fro-norm}, we can select $\O \geq c\beta\k R\max(W,\mu_{0}^2 M)\log^2(MW)$ with appropriate constant $c$ to ensure the desired bound. The result holds with probability $1-O(\k(MW)^{-\beta})$, which follows by using the value of $t$ specified above and then by the union bound over $\k$ independent partitions. 
%----------------------------------------------------------------------------------------------------------------------------------
%Coherence-Iterates  

\subsection{Proof of Lemma \ref{lem:CM-coherence-iterates}}
Let $\w_k$ be as defined in \eqref{eq:CM-W-iterate}, and $\vm{e}_m$ be the length-$M$ standard basis vector with $1$ in the $m$th location. The coherence in \eqref{eq:CM-Wk-coherence} can equivalently be written using trace inner product as 
\begin{align}\label{eq:CM-coher-iter} 
\mu^2_k = \frac{\O}{R}\max_{\o \in \Gamma_k} \sum_{m = 1}^M \<\w_{k},\vm{e}_m\ff_{\o}^*\>^2,
\end{align}
which using iterate relation in \eqref{eq:CM-W-iterate} gives  
\begin{align*}
\mu_k^2 &= \frac{\O}{R}\max_{1 \leq \o \leq\O} \sum_{m = 1}^M \<(\k\PT\cA_k^*\cA_k\PT-\PT)\w_{k-1},\vm{e}_m\ff_{\o}^*\>^2.
\end{align*}
In the rest of the proof, we will be concerned with bounding the summands 
\[
\<(\k\PT\cA_k^*\cA_k\PT-\PT)\w_{k-1},\vm{e}_m\ff_{\o}^*\>,
\]
which can be expanded as 
\begin{align*} 
&\<(\k\PT\cA_k^*\cA_k\PT-\PT)\w_{k-1},\vm{e}_m\ff_{\o}^*\> = \sum_{\op \in \Gamma_k} \k\<\PT(\A_{\op}),\vm{e}_m\ff_{\o}^* \>\<\w_{k-1},\A_{\op}\>-\< \w_{k-1},\vm{e}_m\ff_{\o}^*\>\notag\\
&\quad\quad\quad= \sum_{\op \in \Gamma_k} \k\<\PT(\A_{\op}),\vm{e}_m\ff_{\o}^* \>\<\w_{k-1},\A_{\op}\>-\E\k \<\PT(\A_{\op}),\vm{e}_m\ff_{\o}^* \>\<\w_{k-1},\A_{\op}\>.
\end{align*}
To control the deviation of the above sum, we will use the scalar Bernstein inequality. Let 
\[
Z_{\op} = \k(\<\PT(\A_{\op}),\vm{e}_m\ff_{\o}^* \>\<\w_{k-1},\A_{\op}\>-\E \<\PT(\A_{\op}),\vm{e}_m\ff_{\o}^*\>\<\w_{k-1},\A_{\op}\>).
\]
The variance $\sum_{{\op} \in \Gamma_k} \E Z_\op Z_\op^*$ is upper bounded by 
\begin{align}
\sum_{\op \in \Gamma_k} \E Z_\op Z_\op^* & \leq \k^2\sum_{\op \in \Gamma_k} \E\<\PT(\A_{\op}),\vm{e}_m\ff_{\o}^*\>\<\PT(\A_{\op}),\vm{e}_m\ff_{\o}^*\>^* \<\w_{k-1},\A_{\op}\>\<\w_{k-1},\A_{\op}\>^*\notag \\
& = \k^2\sum_{\op \in \Gamma_k} \E|\<\PT(\A_{\op}),\vm{e}_m\ff_{\o}^*\>|^2 |\<\w_{k-1},\A_{\op}\>|^2\label{eq:CM-var-upperbound}
\end{align}
Let $\bar{\uu}_m^*$ denote the $m$th row of the matrix $\u$. The term $\< \PT \A_{\op},\vm{e}_m\ff_{\o}^*\>$ can be expanded using \eqref{eq:CM-PT-def} as follows:
\begin{align}
\< \PT(\A_{\op}),\vm{e}_m\ff_{\o}^*\> &= \tr{\PT (\A_{\op})\ff_{\o}\vm{e}_m^*}\notag\\
& = \< \u\u^*\vm{d}_{\op}\vm{f}_{\op}^*,\vm{e}_m\ff_{\o}^*\> + \< \vm{d}_{\op}\vm{f}_{\op}^*\v\v^*,\vm{e}_m\ff_{\o}^*\> - \< \u\u^*\vm{d}_{\op}\vm{f}_{\op}^*\v\v^*, \vm{e}_m\ff_{\o}^*\>\notag\\
 &= \<\bar{\uu}_m,\u^*\vm{d}_{\op}\>(\vm{f}_{\op}^*\vm{f}_{\o})+ \<\v^*\ff_{\op},\v^*\vm{f}_{\o}\>d_{\op}[m] - \<\bar{\uu}_m,\u^*\d_{\op}\>\<\v^*\ff_\op,\v^*\ff_\o\>\label{eq:CM-var-quant}
\end{align}
Let $Y_1 = \<\bar{\uu}_m,\u^*\vm{d}_{\op}\>(\vm{f}_{\op}^*\vm{f}_{\o})$, $Y_2 = \<\v^*\ff_{\op},\v^*\vm{f}_{\o}\>d_{\op}[m]$, and $Y_3 = \<\bar{\uu}_m,\u^*\d_{\op}\>\<\v^*\ff_\op,\v^*\ff_\o\>$. Using this notation and combining \eqref{eq:CM-var-upperbound}, \eqref{eq:CM-var-quant}, and expanding the square, it is clear that 
\begin{align}\label{eq:CM-var-interm}
\sum_{\op \in \Gamma_k}\E Z_\op Z_\op^* &\leq \k^2\E\sum_{\op \in \Gamma_k} 3(|Y_1|^2+|Y_2|^2+|Y_3|^2)|\<\w_{k-1},\A_{\op}\>|^2.
\end{align}
Therefore, the term required to calculate the variance are the following: first,
\begin{align*}
\sum_{\op \in \Gamma_k} \E |Y_1|^2|\<\w_{k-1},\A_{\op}\>|^2
 &\leq \bar{\uu}_m^*\u^*\max_{\op}\E( \<\w_{k-1},\A_{\op}\>^2\d_{\op}\d_{\op}^*)\u\bar{\uu}_m\cdot\vm{f}_{\o}^*\sum_{\op \in \Gamma_p}(\ff_\op\ff_{\op}^*)\vm{f}_{\o},
\end{align*}
and the result of Lemma \ref{lem:CM-Supporting1} shows that
 
\begin{align*}
\E( |\<\w_{k-1},\A_{\op}\>|^2\d_{\op}\d_{\op}^*) \preccurlyeq 3\|\w_{k-1}\ff_{\op}\|_2^2 \vm{I}_M.
\end{align*}
Thus,
\begin{align}
\sum_{\op \in \Gamma_k} \E |Y_1|^2|\<\w_{k-1},\A_{\op}\>|^2
  & \leq 3\bar{\uu}_m^*\u^*\u \bar{\uu}_m \left\|\w_{k-1}\ff_{\op}\right\|_2^2 \cdot \frac{1}{\k}\|\ff_\o\|_2^2 \leq 3\|\bar{\uu}_m\|_2^2 \mu_{k-1}^2 \frac{WR}{\k\O^2};\label{eq:CM-subterm1}
\end{align}
second, 
\begin{align*}
\sum_{\op \in \Gamma_k} |Y_2|^2 & =   \vm{f}_{\o}^*\v\v^*\sum_{\op \in \Gamma_k}(\vm{f}_{\op}\vm{f}_{\op}^*)\v\v^*\vm{f}_{\o} = \frac{1}{\k}\|\v^*\vm{f}_{\o}\|_2^2 \leq \mu_0^2\frac{R}{\k\O};
\end{align*}
and hence 
\begin{align}
\E \sum_{\op \in \Gamma_k} |Y_2|^2 |\<\w_{k-1},\A_{\op}\>|^2 &\leq \max_{\op}\E |\<\w_{k-1},\A_{\op}\>|^2 \cdot \sum_{\op \in \Gamma_k} |Y_2|^2\leq \mu_0^2\mu_{k-1}^2 \frac{R^2}{\k\O^2};\label{eq:CM-subterm2}
\end{align}
third, since $|Y_3|^2 = |Y_1|^2|Y_2|^2/|\vm{f}_{\op}^*\vm{f}_{\o}|^2$, we can combine the first two terms to obtain
\begin{align}
\E \sum_{\op \in \Gamma_k} |Y_3|^2 |\<\w_{k-1},\A_{\op}\>|^2 \leq 3\|\bar{\uu}_m\|_2^2 \mu_0^2\mu_{k-1}^2 \frac{R^2}{\k\O^2}.\label{eq:CM-subterm3}
\end{align}
Plugging \eqref{eq:CM-subterm1},\eqref{eq:CM-subterm2}, and \eqref{eq:CM-subterm3} in \eqref{eq:CM-var-interm}, 
\begin{align*}
\sigma_Z^2 &= \sum_{\op \in \Gamma_k} \E Z_\op Z_\op^* \leq  3\k \left(\mu_0^2\mu_{k-1}^2 \frac{R^2}{\O^2} +3\left\|\bar{\uu}_m\right\|_2^2\mu_{k-1}^2\frac{WR}{\O^2}+ 3\left\|\bar{\uu}_m\right\|_2^2\mu_0^2\mu_{k-1}^2 \frac{R^2}{\O^2}\right)\\
&= 3\k \left(4\mu_0^2\mu_{k-1}^2 \frac{R^2}{\O^2} +3\left\|\bar{\uu}_m\right\|_2^2\mu_{k-1}^2\frac{WR}{\O^2}\right),
\end{align*}
where the last inequality follows by using the fact that $\|\bar{\uu}_m\|_2 \leq 1$.  Using $t = \beta\log(MW)$, we obtain the first quantity in the maximum in \eqref{eq:matbernpsi}
\begin{align}\label{eq:CM-variance-Zv}
\sigma_Z^2\beta\log(MW) \leq 3\k \left(4\mu_0^2\mu_{k-1}^2 \frac{R^2}{\O^2} +3\left\|\bar{\uu}_m\right\|_2^2\mu_{k-1}^2\frac{WR}{\O^2}\right)\beta\log(MW).
\end{align}
Now, we will show that the variable
\[ 
Z_{\op} = (Y_1 + Y_2 - Y_3)\<\w_{k-1},\A_{\op}\>
\]
is a subexponential random variable. It is easy to show that 
\[
\|Y_1\|_{\psi_2}^2 \leq c\|\bar{\uu}_m\|_2^2(\ff_{\op}^*\ff_{\o})^2 \leq c\|\bar{\uu}_m\|_2^2 \frac{W^2}{\O^2}, 
\]
\[
\|Y_2\|_{\psi_2}^2 \leq c\<\v^*\ff_{\op},\v^*\vm{f}_{\o}\>^2 \leq c\mu_0^4 \frac{R^2}{\O^2},
\]
and 
\[
\|Y_3\|_{\psi_2}^2 \leq  c\|\bar{\uu}_m\|_2^2\<\v^*\ff_{\op},\v^*\vm{f}_{\o}\>^2 \leq c\|\bar{\uu}_m\|_2^2\mu_0^4 \frac{R^2}{\O^2}.
\]
Then the fact $\|Y_1 + Y_2 -Y_3 \|_{\psi_2} \leq \|Y_1\|_{\psi_2} +\|Y_2\|_{\psi_2}+ \|Y_3\|_{\psi_2}$ implies that the sum $Y_1+Y_2-Y_3$ is also a subgaussian. Using another standard calculation, it can be shown that 
\[
\|\<\w_{k-1},\A_{\op}\>\|_{\psi_2}^2 \leq c\|W_{k-1}\ff_{\op}\|_2^2 \leq c\mu_{k-1}^2\frac{R}{\O}.
\]
It is shown in Lemma \ref{lem:CM-Orlicz-1and2} that product $X$ of two subgaussian random variables $X_1$, and $X_2$ is subexponential and $\|X\|_{\psi_1} \leq c\|X_1\|_{\psi_2}\|X_2\|_{\psi_2}$. This fact now implies that $Z_{\op}$ is a subexponential random variable with Orlicz-1 norm 
\begin{align*}
\left\|Z_{\op}\right\|_{\psi_1}^2 &\leq \k^2\mu_0^4\mu_{k-1}^2\frac{R^3}{\O^3}+ 3\k^2\mu_{k-1}^2\frac{W^2R}{\O^3}\|\bar{\vm{u}}_m\|_2^2+ 3\|\bar{\vm{u}}_m\|_2^2\k^2\mu_0^4\mu_{k-1}^2\frac{R^3}{\O^3}\\
& \leq 4\k^2\mu_0^4\mu_{k-1}^2\frac{R^3}{\O^3}+3\k^2\mu_{k-1}^2\frac{W^2R}{\O^3}\|\bar{\vm{u}}_m\|_2^2,
\end{align*}
where the last inequality follows from $\|\bar{\vm{u}}_m\|_2^2 \leq 1$.
Choosing $t = \beta\log(MW)$, as before, gives the second quantity in the maximum in \eqref{eq:matbernpsi}
\begin{equation}\label{eq:CM-orlicz-norm-Zv}
U_1^2\log^2 \left(|\Gamma_k|\frac{U_1^2}{\sigma_Z^2}\right)\beta^2\log^2(MW) \leq \k^2\mu_{k-1}^2\frac{4\mu_0^4R^3+3\|\bar{\vm{u}}_m\|_2^2W^2R}{\O^3}\beta^2\log^4(MW).
\end{equation}
Using Bernstein bound, it follows that $|\<\w_k,\vm{e}_m\ff_\o^*\>|$ is dominated by the maximum of \eqref{eq:CM-variance-Zv}, and \eqref{eq:CM-orlicz-norm-Zv} with probability at least $1-(MW)^{-\beta}$. Using this bound in \eqref{eq:CM-coher-iter}, and using the fact that $\sum_{m = 1}^M \|\bar{\uu}_m\|_2^2 = R$, we obtain the following bound on $\mu_k^2$ with probability (using the union bound) at least $1-O(|\Gamma_k|(MW)^{-\beta})$ 
\[
\mu_{k}^2 \leq c\mu_{k-1}^2\max\left\{3\k \frac{4\mu_0^2MR +3\mu_{k-1}^2 WR}{\O}\beta\log(MW),\k^2\frac{4\mu_0^4MR^2+3W^2R}{\O^2}\beta^2\log^4(MW)\right\}.
\]
Now taking $\O \geq c\beta\k R(\mu_0^2M+W)\log^2(MW)$ gives us the desired bound on the coherence $\mu_{k}^2$ for a fixed value of $k$ with probability $1-O(|\Gamma_k|(MW)^{-\beta})$. Using union bound over $\k$ independent partitions, the failure probability becomes $1-O(\O(MW)^{-\beta})$.

\begin{lem}\label{lem:CM-Supporting1}
Let $\vm{d}_{\o} \in \{-1,1\}^M$ denote the binary length-$M$ random vectors as defined in \eqref{eq:CM-Aomega-def}. Then 
\[
\E |\<\vm{C},\A_{\o}\>|^2 \vm{d}_{\o}\vm{d}_{\o}^*  \preccurlyeq 3\|\vm{C}\ff_{\o}\|_2^2 \vm{I}_M
\]
\end{lem}
\begin{proof}
Let $\{\vm{c}_m^*\}_{1 \leq m \leq M}$ denote the rows of the matrix $\vm{C}\in \comps^{M\times W}$, $\{\x\}_{(\alpha,\beta)}$ denote the $(\alpha,\beta)$th entry of $\x$, and $\vm{A}_{\o}$ as defined in \eqref{eq:CM-Aomega-def}. Then we can write
\begin{align*} 
\{\E( |\<\vm{C},\A_{\o}\>|^2\d_{\o}\d_{\o}^*)\}_{(\alpha,\beta)} &= \E \left|\sum_{m = 1}^M d_{\o}[m] \vm{c}_m^*\ff_{\o}\right|^2\{\d_{\o}\d_{\o}^*\}_{(\alpha,\beta)}\\
&= \sum_{m = 1}^M  |\vm{c}_{\alpha}^*\ff_{\o}|^2\delta_{\alpha = \beta}  + 2\<\vm{c}_{\alpha}^*\ff_{\o},\vm{c}_{\beta}^*\ff_{\o}\>\delta_{\alpha \neq \beta},
\end{align*}
where $\delta_{\alpha = \beta}$ is $1$ when $\alpha = \beta$ and is $0$ otherwise. Similarly  $\delta_{\alpha \neq \beta}$ is $1$ when $\alpha \neq \beta$ and is $0$ otherwise. This implies that 
\begin{align*}
\E( |\<\vm{C},\A_{\o}\>|^2\d_{\o}\d_{\o}^*) &= \left\|\vm{C}\ff_{\o}\right\|_2^2 \vm{I}_M+ 2\vm{C}\ff_{\o}\ff_{\o}^*\vm{C}^* -2 \mbox{diag}(\vm{C}\ff_{\o}\ff_{\o}^*\vm{C}^*)\\
& \prec \left\|\vm{C}\ff_{\o}\right\|_2^2 \vm{I}_M+ 2\vm{C}\ff_{\o}\ff_{\o}^*\vm{C}^* \preccurlyeq 3\|\vm{C}\ff_{\o}\|_2^2 \vm{I}_M
\end{align*}
where the first inequality follows from the fact that $\mbox{diag}(\vm{C}\ff_{\o}\ff_{\o}^*\vm{C}^*)$ is a positive-semidefinite matrix, and the last inequality is valid because for a vector $\vm{x}$, we have $\left\|\vm{x}\right\|_2^2\vm{I} \succcurlyeq \vm{x}\vm{x}^*$.
\end{proof}
\begin{lem}[Lemma 5.9 in \cite{vershynin10in}]\label{lem:CM-Orlicz-independence}
Consider a finite number $Q$ of independent subgaussian random variable $X_q$. Then, 
\[
\left\|\sum_{q = 1}^Q X_q\right\|_{\psi_2}^2 \leq c\sum_{q =1 }^Q \|X_q\|_{\psi_2}^2,
\]
where $c$ is an absolute constant.
\end{lem}
\begin{lem}[Lemma 5.14 in \cite{vershynin10in}]\label{lem:CM-Orlicz-square} A random variable $X$ is subgaussian iff $X^2$ is subexponential. Furthermore, 
\[
\|X\|_{\psi_2}^2 \leq \|X^2\|_{\psi_1} \leq 2\|X\|_{\psi_2}^2.
\]
\end{lem}
\begin{lem}\label{lem:CM-Orlicz-1and2}
Let $X_1$, and $X_2$ be two subgaussian random variables, i.e., $\|X_1\|_{\psi_2} < \infty$, and $\|X_2\|_{\psi_2} < \infty$. Then the product $X_1X_2$ is a subexponential random variable with 
\[
\|X_1X_2\|_{\psi_1} \leq c\|X_1\|_{\psi_2}\|X_2\|_{\psi_2}.
\]
\end{lem}
\begin{proof}
For a subgaussian random variable, the tail behavior is
\[
\P{|X| > t} \leq \er\cdot\exp\left(\frac{-ct^2}{\|X\|_{\psi_2}^2}\right)\quad \forall t >0;
\]
see, for example, \cite{vershynin10in}. We are interested in 
\begin{align*}
\P{|X_1X_2| > \lambda} &\leq \P{|X_1| > t}+\P{|X_2| > \lambda/t}\\
& \leq \er\cdot\exp\left(-ct^2/\|X_1\|_{\psi_2}^2\right)+\er\cdot\exp\left(-c\lambda^2/t^2\|X_2\|_{\psi_2}^2\right).
\end{align*}
Select $t^2 = \lambda \|X_1\|_{\psi_2}/\|X_2\|_{\psi_2}$, which gives 
\[
\P{|X_1X_2| > \lambda} \leq 2\er\cdot\exp\left(-c\lambda/\|X_1\|_{\psi_2}\|X_2\|_{\psi_2}\right).
\]
Now Lemma 2.2.1 in \cite{vandervaart96we} implies that if a random variable $Z$ obeys $\P{|Z| > u} \leq \alpha \er^{-\beta u}$, then $\|Z\|_{\psi_1} \leq (1+\alpha)/\beta$. Using this result, we obtain 
\[
\|X_1X_2\|_{\psi_1} \leq c\|X_1\|_{\psi_2}\|X_2\|_{\psi_2},
\]
which proves the result. 
\end{proof}

\section{Proof of Theorem \ref{thm:CM-stablerec-CM1}: Stability of the M-Mux}
\label{sec:CM-Theory2}

Given the contaminated measurements, as in \eqref{eq:CM-noisy-meas}, and the linear operator $\cA^*$, which is the adjoint $\cA$, defined in \eqref{eq:CM-M-Mux-meas}, we have 
\begin{align}
\|\cA^*(\vm{y})-\E\cA^*(\vm{y})\| &\leq \|(\cA^*\cA-\mathcal{I})(\vm{C}_0)\|+\|\cA^*(\vm{\xi})\|\notag\\
&= \theta_1 + \theta_2\label{eq:CM-Lambda-bound}
\end{align}
The result of Theorem \ref{thm:CM-stablerec-CM1} can be considered as the corollary of the following result in \cite{koltchinskii10nu}. 
\begin{thm}\label{thm:CM-KLTthm}\cite{koltchinskii10nu}
Let $\tilde{\vm{C}} \in \comps^{M \times W}$ be the estimate of rank-$R$ matrix $\vm{C}_0$, defined in \eqref{eq:CM-C0-def}, from the measurements $\vm{y}$ in \eqref{eq:CM-noisy-meas} using the estimator in \eqref{eq:CM-KLT-est}. If $\lambda \geq 2\|\cA^*\vm{y}\|$, then
\begin{equation}\label{eq:CM-KLTbound}
\|\tilde{\vm{C}}-\vm{C}_0\|_{\F}^2 \leq \min \{ 2\lambda\|\vm{C}_0\|_*, 1.5\lambda^2 R\}
\end{equation}
\end{thm}
To prove Theorem \ref{thm:CM-stablerec-CM1}, we only need to compute a bound on the operator norm in \eqref{eq:CM-Lambda-bound}. The bound on $\theta_1$ in \eqref{eq:CM-Lambda-bound} is provided by the following corollary of Lemma \ref{lem:CM-concentration}. With out loss of generality, we will assume that $\|\vm{C}_0\|_{\F} = 1$.
\begin{cor}\label{cor:CM-concentration}
Let $\mu_{0}^2$, defined in \eqref{eq:CM-coherence}, be the coherence of rank-$R$ matrix $\vm{C}_0$ in \eqref{eq:CM-C0-def}. Then for all $\beta \geq 1$ 
$$\|\cA^*\cA(\vm{C}_0) - \vm{C}_0\|^2 \leq c\sqrt{\frac{\beta \max(\mu_0^2 M,W)\log(MW)}{\O}}\left\|\vm{C}_0\right\|_{\F}$$
with probability at least $1-(MW)^{-\beta}$ provided $\O \geq c\beta\min(\mu_0^2 M,W) \log^2(MW)$. 
\end{cor}
The proof of the corollary follows from Lemma \ref{lem:CM-concentration}. In particular, the corollary is a direct result of the bound \eqref{eq:CM-concentration-bound} by taking $k = 1$. The first term in \eqref{eq:CM-concentration-bound} dominates when $\O \geq c\beta\min(\mu_0^2 M,W) \log^2(MW)$. 

The upper bound on $\theta_2$ follows from the following Lemma. 
\begin{lem}\label{lem:CM-theta2bound}
Let $\cA^*: \reals^{\O} \rightarrow \comps^{M\times W}$be the adjoint of the linear operator $\cA$ defined in \eqref{eq:CM-M-Mux-meas}, and $\vm{\xi}$ be the noise random variable with statistics given in \eqref{eq:CM-noise-stats}, and $\|\vm{\xi}\|_{\psi_2} \leq \delta$. Then for $\beta \geq 1$, the conclusion: 
$$\|\cA^*(\vm{\xi})\|^2 \leq c \|\vm{\xi}\|_{\psi_2}\sqrt{\frac{\beta\max(W,M)\log(MW)}{\O}}$$
holds with probability at least $1-(MW)^{-\beta}$, when $\O \geq c\beta\min(W,M)\log^2(MW)$.
\end{lem}
%The quantity in \eqref{eq:CM-Lambda-bound} will be bounded using matrix Bernstein inequality in Proposition \ref{prop:matbernpsi}. 
Combining the above bounds with \eqref{eq:CM-Lambda-bound} gives
\begin{equation}\label{eq:CM-lambda-bound} 
\|\cA^*(\vm{y})-\E\cA^*(\vm{y})\| \leq c\sqrt{\frac{\beta\{\max(W,\mu_0^2 M)+\|\vm{\xi}\|_{\psi_2}^2\max(W,M)\}\log(MW)}{\O}}
\end{equation}
with high probability. The second term is meaningful in the minimum in \eqref{eq:CM-KLTbound} in Theorem \ref{thm:CM-KLTthm} when we select the sampling rate $\O$ large enough that makes $\lambda^2 \ll 1$. Theorem \ref{thm:CM-KLTthm}, and \eqref{eq:CM-lambda-bound} assert that
\[ 
\|\tilde{\vm{C}}-\vm{C}_0\|_{\F}^2 \leq c\|\vm{\xi}\|_{\psi_2} \leq c\delta,
\]
when $\O \geq c\beta R\max(W,\mu_0^2M)\log^2(MW)$, which does not violate the upper bounds on $\Omega$ in Corollary \ref{cor:CM-concentration}, and Lemma \ref{lem:CM-theta2bound}. This proves Theorem \ref{thm:CM-stablerec-CM1}.
\subsection{Proof of Lemma \ref{lem:CM-theta2bound}}
We will use the orlicz version of the matrix Bernstein's inequality \ref{prop:matbernpsi}. 
\begin{proof}
We are interested in bounding $\cA^*(\vm{\xi}) = \sum_{\o = 1}^\O \xi[\o]\A_\o$. Let $\z_\o = \xi[\o]\A_\o$. It is clear that $\E\z_\o = \vm{0}$, which follows by the independence of $\xi[\o]$, and $\A_\o$, and by the fact that $\E\xi[\o] = 0$. To use the Bernstein bound, we need to calculate the variance \eqref{eq:matbernsigma}. We begin with
\begin{align*} 
\|\sum_{\o = 1}^\O \E\z_\o\z_\o^*\| &= \left\|\E\sum_{\o = 1}^\O  \xi[\o]^2 \ff_\o\d_\o^*\d_\o\ff_\o^*\right\| =  \left\|\E\sum_{\o=1}^\O \xi[\o]^2\|\d_\o\|_2^2\ff_\o\ff_\o^*\right\|\\
& \leq M\max_\o \E\xi[\o]^2 \left\|\sum_{\o=1}^\O\ff_\o\ff_\o^*\right\| = M\max_{\o}\|\xi[\o]\|_{\psi_2}^2  \leq c\frac{M}{\O} \|\vm{\xi}\|_{\psi_2}^2, ~~~ \mbox{Using}~\eqref{eq:CM-noise-stats}
\end{align*}
Similarly, 
\begin{align*} 
\left\|\sum_{\o = 1}^\O \E\z_\o^*\z_\o\right\| &= \left\|\sum_{\o = 1}^\O\E\xi[\o]^2 \ff_\o^*\E (\d_\o\d_\o^*)\ff_\o\right\| \leq \max_{\o}\|\xi[n]\|_{\psi_2}^2\sum_{\o=1}^\O \|\ff_\o\|_2^2\\
&\leq c\frac{W}{\O}\|\vm{\xi}\|_{\psi_2}^2, ~~~ \mbox{Using}~\eqref{eq:CM-noise-stats}
\end{align*}
Then, we obtain 
\[
\sigma_Z^2 \leq c\|\vm{\xi}\|_{\psi_2}^2\frac{\max(W,M)}{\O}.
\]
Since $\|\z_\o\| = |\xi[\o]|\|\A_\o\| \leq |\xi[\o]|(MW)/\O$, we have
\begin{align*}
\|\z_\o\|_{\psi_2} & \leq c\|\vm{\xi}\|_{\psi_2}\sqrt{\frac{MW}{\O^2}}.
\end{align*}
Thus, 
\[
U_2\log^{1/2}\left(\frac{\O U_2^2 }{\sigma_Z^2}\right) \leq c\|\vm{\xi}\|_{\psi_2}\sqrt{\frac{MW}{\O^2}}\log^{1/2}(MW).
\]
Now using $t = \beta\log(MW)$, we obtain 
\[
\|\cA^*(\vm{\xi})\| \leq c\|\vm{\xi}\|_{\psi_2}\max\left\{\sqrt{\frac{\beta\max(W,M)\log(MW)}{\O}}, \sqrt{\beta^2\frac{MW\log^3(MW)}{\O^2}}\right\}
\]
with probability at least $1-(MW)^{-\beta}$. The first term in the minimum dominates when $\O \geq c\beta \min(W,M)\log^2(MW)$. This proves the Lemma.  
\end{proof}

\section{Proof of Theorem \ref{thm:CM-exactstablerec-CM2}: Matrix RIP for the FM-Mux}
\label{sec:CM-Theory3}

In this section, we will establish the matrix RIP for the operator $\cB$ defined in \eqref{eq:CM-PhiD}. The measurements in \eqref{eq:CM-PhiD} can be expressed as
\begin{align}
\vm{y} = \cB(\vm{C}_0) = \vm{\Phi}\vm{D}\cdot\mbox{vec}(\vm{C}_0\tilde{\vm{F}})\label{eq:CM-cB-def2},
\end{align}
where $\vm{\Phi} = [\vm{H}_1, \cdots, \vm{H}_M]$ is a block-circulant matrix, and $\vm{D}: \O M \times \O M$ is a large diagonal matrix formed by cascading smaller $\O \times \O$ diagonal matrices $\{\vm{D}_m\}_{ 1\leq m\leq M}$, defined earlier, along the diagonal. The proof of Theorem \ref{thm:CM-exactstablerec-CM2} is then just a combination of three existing results in the literature.

\begin{enumerate}	
	\item In \cite{candes11ti}, it is shown the matrix RIP for an operator $\mathcal{T}:\C^{M\times W}\rightarrow\R^\Omega$ follows immediately from establishing a concentration inequality.  In particular, if for any fixed $M \times W$ matrix $\vm{C}$,
\begin{equation}\label{eq:Yaniv-result1}
		\P{\left| \|\mathcal{T}(\vm{C})\|_2^2 - \|\vm{C}\|_{\F}^2 \right| > \delta\|\vm{C}\|_{\F}^2}
		~\leq~
		2\er^{-\Omega/t}
\end{equation} 
	for $\delta = 0.3/2$, then the linear operator $\mathcal{T}$ satisfies the low-rank RIP when 
\begin{equation}\label{eq:Yaniv-result2}
		\Omega ~\gtrsim~ t R(W+M) \Rightarrow \delta_{2R}(\cB) ~\leq~0.3,
\end{equation}
with probability at least $1-c\er^{-d \O}$ for fixed constants $c,d >0$, and an appropriately chosen $t$ that depends on $\delta$.  
	
	\item In \cite{krahmer2010new}, it is shown that if a matrix $\vm{\Phi}$ obeys the sparse RIP, 
	\[
		0.85\|\vm{x}\|_2^2 ~\leq~ \|\vm{\Phi}\vm{x}\|_2^2 ~\leq~ 1.15\|\vm{x}\|_2^2,
	\]
	for all length-$\O M$, $K$-sparse vectors $\vm{x}$, then for an arbitrary fixed $\vm{x}$, the matrix $\vm{\Phi}\vm{D}$ obeys the concentration inequality
	\[
		\P{\left|\|\vm{\Phi}\vm{D}\vm{x}\|_2^2 - \|\vm{x}\|_2^2\right| > 0.15\|\vm{x}\|_2^2} ~\leq~
		2\er^{-K/c_1}
	\]
	for a fixed constant $c_1 >0$. We can just as well take $\vm{x} = \vc(\vm{C}\tilde{\vm{F}})$ for a fixed $M \times W$ matrix $\vm{C}$ to obtain 
	\begin{align*}
\P{|\|\mathbf{\Phi}\vm{D}\vc(\vm{C}\tilde{\vm{F}})\|_2^2-\|\vc(\vm{C}\tilde{\vm{F}})\|_2^2|> 0.15\|\vc(\vm{C}\tilde{\vm{F}})\|_2^2} \leq 2\er^{-K/c_1}.
\end{align*}
Obviously, $\|\vc(\vm{C}\tilde{\vm{F}})\|_2^2 = \|\vm{C}\tilde{\vm{F}}\|_{\F}^2$, and using the fact that the rows of $\tilde{\vm{F}}$ are orthonormal vectors, we have $\|\vm{C}\tilde{\vm{F}}\|_{\F}^2 = \|\vm{C}\|_{\F}^2$, and by the definition of $\cB$, we have the concentration inequality for the linear operator $\cB$
\begin{align}\label{eq:conc1}
\P{|\|\cB(\vm{C})\|_2^2-\|\vm{C}\|_{\F}^2|\geq 0.15\| \vm{C}\|_{\F}^2} \leq 2\er^{-K/c_1}.
\end{align}

\item In \cite{romberg2010sparse,romberg09mu}, the sparse RIP for $\O \times M\O$ random matrix $\vm{\Phi}$ was established for all length-$\O M$, $K$-sparse vectors $\vm{x}$ when  
	\[
		K~\leq~ c_2\Omega/\log^4(M\Omega)
	\]
 with probability at least $1 - c\er^{-d\O}$ for fixed constants $c,c_1,d >0$. This means
\begin{align}\label{eq:conc2}
\P{|\|\cB(\vm{C})\|_2^2-\|\vm{C}\|_{\F}^2|\geq 0.15\| \vm{C}\|_{\F}^2} \leq 2 \er^{-\Omega /t\log^{4}(M\Omega)},
\end{align} 
where $t = c_2/c_1$ and depends on the isometry constant $\delta$ in \eqref{eq:Yaniv-result1}.
 \item Combining the concentration result in \eqref{eq:conc2} with \eqref{eq:Yaniv-result1}, \eqref{eq:Yaniv-result2}, and taking  $\Omega \geq t\beta R(W+M)\log^5(M\O)$ establishes the matrix RIP, which proves Theorem \ref{thm:CM-exactstablerec-CM2}.
\end{enumerate}

In a very similar manner, we can also prove an RIP result and the sampling theorem for the FM-Mux in Figure \ref{fig:CM-FM-Mux2}. The measurements $\vm{y} \in \reals^\O$ in $t \in [0,1]$ can be written as 
\begin{align*}
\vm{y} = [\vm{D}_1\vm{H}_1, \ldots, \vm{D}_M\vm{H}_M]\cdot\mbox{vec}(\vm{C}_0\vm{F}),
\end{align*}
where the matrix $\vm{F}$ now represents a $W \times W$ DFT matrix and, as before, the $\tilde{\vm{F}}$ is the $W \times \O$ partial DFT matrix. The $\O \times \O$ matrices $\{\vm{D}_m\}_{m = 1}^M$ are for modulators but unlike the previous case the circulant filter matrices are now 
\[
\vm{H}_m = \tilde{\vm{F}}^*\hat{\vm{H}}_m\vm{F}, ~ m = 1, \ldots, M,
\]
where $\{\hat{\vm{H}}_m\}_{m = 1}^M$, as before, are $W \times W$ independent diagonal matrices containing independent subgaussian random variables along the diagonal. Define 
\[
\vm{\Phi} = [\vm{D}_1\tilde{\vm{F}}^*, \ldots, \vm{D}_M\tilde{\vm{F}}^*].
\] 
The results in \cite{romberg2010sparse,romberg09mu} also imply that the matrix $\vm{\Phi}$ above obeys an RIP property for sparse vectors, which can be extended to a concentration result when the columns of $\vm{\Phi}$ above are modulated by the independent random variables in the diagonal matrices $\{\hat{\vm{H}}_m\}_{m = 1}^M$. The concentration result then yields a low-rank RIP exactly as before, which says that the FM-Mux in Figure \ref{fig:CM-FM-Mux2} successfully reconstructs the signal ensemble when the ADC is operated at a rate $\O \sim R(W+M)\log^5(MW)$ samples per second. 

\bibliographystyle{IEEEtran}
\bibliography{CMux-references}

\end{document}